\documentclass[acmsmall]{acmart}

\usepackage[english]{babel}

\usepackage{enumitem}
\usepackage{amsmath}

\usepackage{mathtools}
\usepackage{amssymb}
\usepackage{amsthm}
\usepackage{listings}
\usepackage{array}
\usepackage{xspace} 
\usepackage[normalem]{ulem} 
\usepackage{graphicx}
\usepackage{listings}
\graphicspath{ {./figures/} }
\usepackage{varwidth}
\usepackage[linesnumbered,vlined,ruled,commentsnumbered]{algorithm2e}
\usepackage{enumitem}
\usepackage{footnote}
\usepackage{tablefootnote}
\usepackage{float}
\usepackage{caption}
\usepackage{subcaption}
\usepackage{multirow}
\usepackage{makecell}
\newtheorem{definition}{Definition}

\renewcommand{\Re}{\mathbb{R}}

\newcommand{\newAlg}{\textsf{MFD}}
\newcommand{\newStreamAlg}{\textsf{StreamMFD}}

\newcommand{\eps}{\varepsilon}

\newcommand{\expect}{\mathbb{E}}
\newcommand{\fairDiv}{FairDiv}
\newcommand{\streamFairDiv}{SFairDiv}
\newcommand{\queryFairDiv}{QFairDiv}

\newcommand{\stopping}{g}
\def\mparagraph#1{\par\noindent\textbf{#1.}\quad}
\def\eparagraph#1{\par\noindent\textbf{\textit{#1.}}\quad}
\newcommand{\Alg}{$\mathrm{Alg}$}

\definecolor{reva}{rgb}{0.7,0,3}
\definecolor{revb}{rgb}{1,0,0}
\definecolor{revc}{rgb}{1,0.5,0}
\definecolor{revmeta}{rgb}{0,0,1}%{0,0.45,0.05}
\definecolor{default}{rgb}{0,0,0}

\newcommand{\ansmeta}[1]{{{\color{default}{#1}}}}

\newcommand{\ansb}[1]{{{\color{default}{#1}}}}
\newcommand{\ansc}[1]{{{\color{default}{#1}}}}

\usepackage{verbatim}
\usepackage{tcolorbox}
\newtheorem{theorem}{Theorem}[section]
\newtheorem{lemma}[theorem]{Lemma}

\newtheorem{corollary}[theorem]{Corollary}

\newtheorem*{remark*}{Remark}

\newcounter{lpcounter}
\newcommand{\lpproblem}{%
  \refstepcounter{lpcounter}%
  \text{($\mathsf{LP}$\arabic{lpcounter})}%
  \label{lp:\thelpcounter}% Add a label to the problem for referencing
}

\newcounter{fpcounter}
\newcommand{\fpproblem}{%
  \refstepcounter{fpcounter}%
  \text{($\mathsf{FP}$\arabic{fpcounter})}%
  \label{fp:\thelpcounter}% Add a label to the problem for referencing
}

\title{Faster Algorithms for Fair Max-Min Diversification in $\Re^d$}
\titlenote{This work was partially supported by CAHSI-Google IRP.}

\author{Yash Kurkure}
\authornote{The first three authors contributed equally to this research.}
\affiliation{%
  \institution{Department of Computer Science, University of Illinois at Chicago}
  \city{Chicago}
  \country{USA}}
\email{ykurku2@uic.edu}

\author{Miles Shamo}
\authornotemark[2]
\affiliation{%
  \institution{Department of Computer Science, University of Illinois at Chicago}
  \city{Chicago}
  \country{USA}}
\email{milesshamo@gmail.com}

\author{Joseph Wiseman}
\authornotemark[2]
\affiliation{%
  \institution{Department of Computer Science, University of Illinois at Chicago}
  \city{Chicago}
  \country{USA}}
\email{jwisem6@uic.edu}

\author{Sainyam Galhotra}
\affiliation{%
  \institution{Department of Computer Science, Cornell University}
  \city{Ithaca}
  \country{USA}}
\email{sg@cs.cornell.edu}

\author{Stavros Sintos}
\affiliation{%
  \institution{Department of Computer Science, University of Illinois at Chicago}
  \city{Chicago}
  \country{USA}}
\email{stavros@uic.edu}
\renewcommand{\shortauthors}{Yash Kurkure et al.}

\begin{abstract}
The task of extracting a diverse subset from a dataset, often referred to as maximum diversification, plays a pivotal role in various real-world applications that have far-reaching consequences. In this work, we delve into the realm of fairness-aware data subset selection, specifically focusing on the problem of selecting a diverse set of size $k$ from a large collection of $n$ data points (\fairDiv).

The \fairDiv\ problem is well-studied in the data management and theory community. In this work, we develop the first constant approximation algorithm for \fairDiv\ that runs in near-linear time using only linear space. In contrast, all previously known constant approximation algorithms run in super-linear time (with respect to $n$ or $k$) and use super-linear space. Our approach achieves this efficiency by employing a novel combination of the Multiplicative Weight Update method and advanced geometric data structures to implicitly and approximately solve a linear program. Furthermore, we improve the efficiency of our techniques by constructing a coreset. Using our coreset, we also propose the first efficient streaming algorithm for the \fairDiv\ problem whose efficiency does not depend on the distribution of data points. Empirical evaluation on million-sized datasets demonstrates that our algorithm achieves the best diversity within a minute. All prior techniques are either highly inefficient or do not generate a good solution.
\end{abstract}



\begin{document}

\maketitle

\section{Introduction}
\label{sec:intro}
In numerous real-world scenarios, including data summarization, web search, recommendation systems, and feature selection, it is imperative to extract a diverse subset from a dataset (often referred to as maximum diversification). The decisions made in these domains have significant consequences. Therefore, it is crucial to guarantee that the outcomes are not only diverse but also unbiased. For instance, while a primary objective of data summarization is to select a representative sample that encapsulates analogous data points, conventional summarization techniques have been identified to exhibit biases against minority groups, leading to detrimental repercussions. In this work, we study the problem of fairness-aware maximum diversification, where the goal is to choose a diverse set of representative data points satisfying a group fairness constraint.

%Data-driven algorithms are increasingly being used to support decision-making in many areas of human activity, and thereby impacting the lives of millions of individuals. An important challenge that arises in such systems is ensuring \emph{diversity} and \emph{fairness}. Ideally, algorithms that are used for real-applications should not only return accurate solutions, but also solutions that are fair to different minority groups in the input data. For example, it is critical for algorithms that are used in applications of urban planning~\cite{irschik2016vienna, jalkanen2020analyzing}, college admissions~\cite{schwartz2004fair, zwick2017gets, bhattacharya2017university}, voting, banking, and criminal justice~\cite{chouldechova2020snapshot, olteanu2019social} to ensure fairness. In databases and machine learning there has been a lot of research on algorithmic fairness studying problems in classification~\cite{dwork2012fairness}, clustering~\cite{chierichetti2017fair}, ranking~\cite{narasimhan2020pairwise, zehlike2017fa}, matching~\cite{sankar2021matchings}, data summarization~\cite{celis2018fair, kleindessner2019fair}, etc.

We consider the problem of ensuring group fairness in max-min diversification. We are given a set of items\footnote{The terms item and points are used interchangeably.}, where each item belongs to one group determined by a sensitive attribute (we refer to it as color). Given a parameter $k_i$ for each color $i$, the goal is to return a subset of items $S$ such that, $S$ contains at least $k_j$ points from each group $j$, and the minimum pairwise distance in $S$ is maximized.
We study the problem in the geometric setting where input items are points in $\Re^d$, for a constant dimension $d$. This setting encompasses the majority of realistic scenarios because many datasets are represented as points in the Euclidean space. Even if the input items are not points in $\Re^d$, it is often the case that 
the items can be embedded (with low error) in a geometric space with low intrinsic dimension~\cite{verbeek2014metric, yang2012defining, pope2021intrinsic}.
Although we focus on the Euclidean space, some of our algorithms can be extended to metric spaces with bounded doubling dimension~\cite{har2005fast, beygelzimer2006cover, feldmann2020parameterized, ng2002predicting}.

\begin{figure}
     \centering
     
     \begin{subfigure}[b]{0.45\textwidth}
         \centering
         \includegraphics[width=0.5\textwidth]{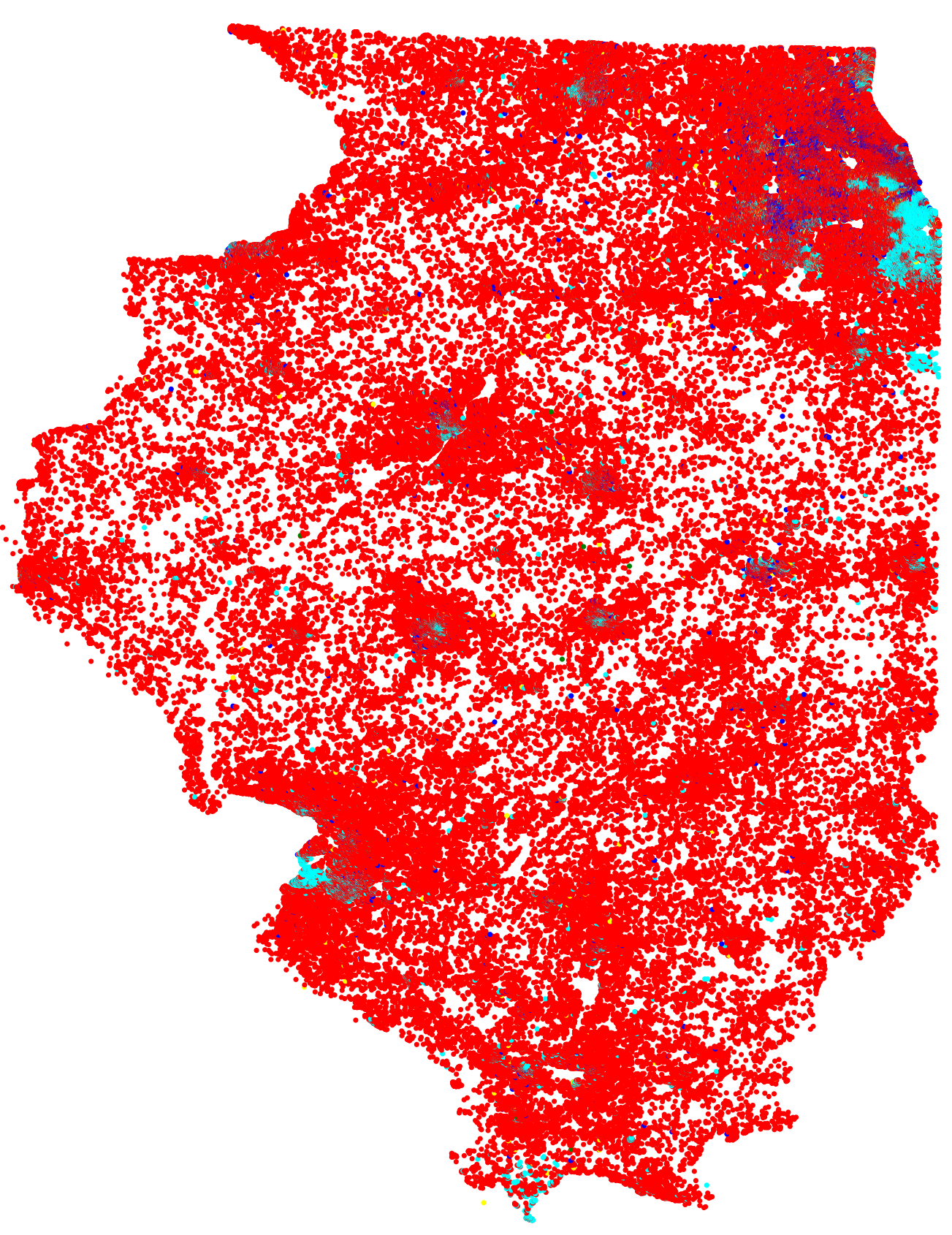}
         \caption{Individuals represented on the map}
         \label{}
     \end{subfigure}
     \hfill
     \begin{subfigure}[b]{0.45\textwidth}
         \centering
         \raisebox{\height}{\includegraphics[width=0.8\textwidth]{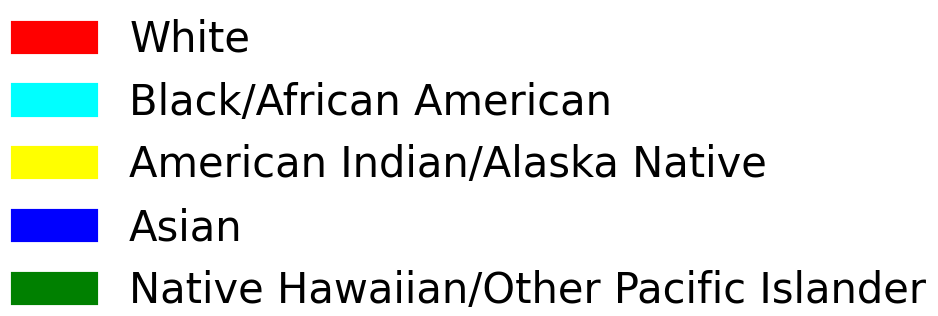}}
         \label{fig:five over x}
     \end{subfigure}\\
     \begin{subfigure}[b]{0.45\textwidth}
         \centering
         \includegraphics[width=0.5\textwidth]{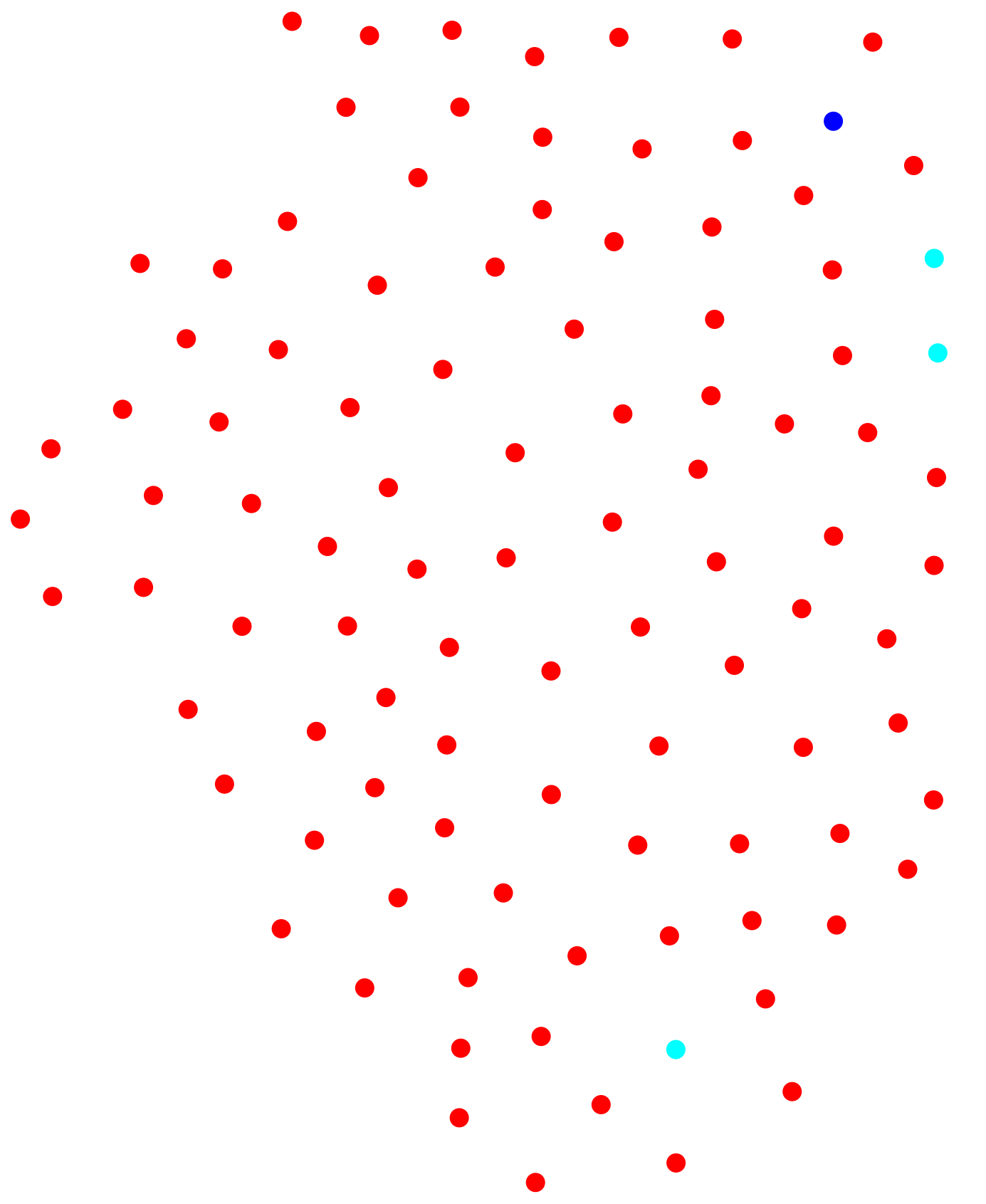}
         \caption{Max-min diversification without fairness constraints}
         \label{}
     \end{subfigure}
     \hfill
     \begin{subfigure}[b]{0.45\textwidth}
         \centering
         \includegraphics[width=0.5\textwidth]{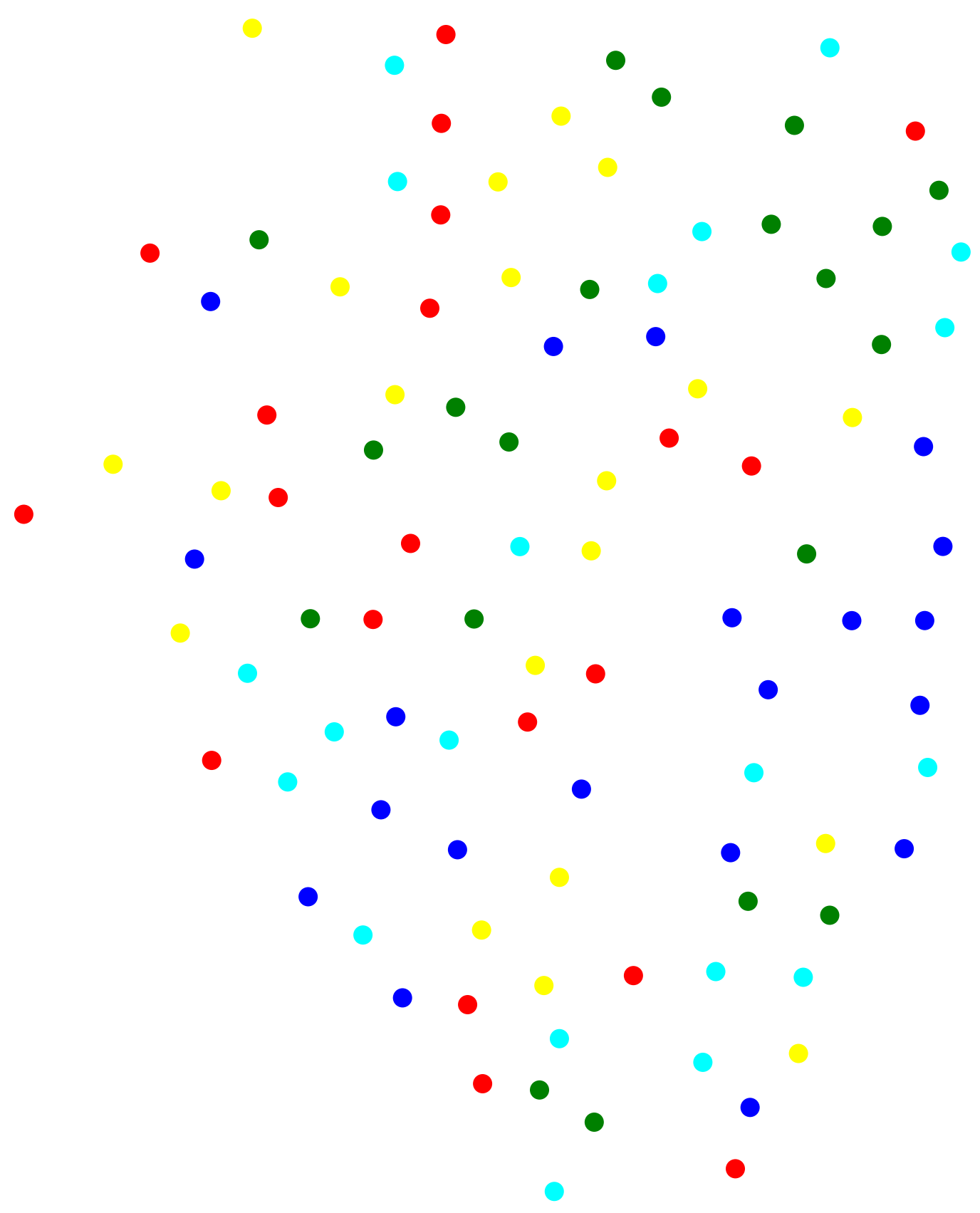}
         \caption{Max-min diversification with fairness constraints}
         \label{fig:five over x}
     \end{subfigure}
     
     \vspace{-3mm}
        \caption{Example scenario where each point denotes an individual in the state of Illinois. (b) shows the biased output of max-min diversification technique on this dataset. (c) denotes the fair output of our algorithm. }
        \label{fig:example}
        \vspace{-3mm}
\end{figure}

\begin{comment}
\begin{figure}[H]
\centering
\includegraphics[scale=0.2]{plots/mmd_vs_fmmd/original.png}
\end{figure}
\begin{figure}[H]
\centering
\includegraphics[scale=0.2]{plots/mmd_vs_fmmd/mmd.png}
\end{figure}
\begin{figure}[H]
\centering
\includegraphics[scale=0.2]{plots/mmd_vs_fmmd/fmmd.png}
\end{figure}
\end{comment}

We motivate the problem with the following example.
\begin{example}
Consider a state-court of Illinois, which wants to form a jury consisting of individuals from the state. One of the primary goals of this jury selection task is to identify individuals from neighborhoods that are far apart, i.e. maximize the minimum distance between selected individuals and have representation of people from diverse backgrounds and cultures. This problem has often been studied as a max-min diversification problem~\cite{fukurai1991cross}.
Figure~\ref{fig:example} (a) shows the different individuals on the map, where each point is colored based on their sensitive attribute. Running a traditional Max-min diversification algorithm on this dataset returned highly biased results (most of the returned points belong to white, as shown in Figure~\ref{fig:example} (b)). Using the fairness constraint, the output contains points from different sensitive groups. The issue of fair jury selection has been a key focus of courts across the world~\cite{url1,url2}, where many studies have recorded biases in jury and its consequences in their decisions~\cite{gau2016jury}.
\end{example}

This example motivates the importance of studying fairness aware variant of Max-Min diversification. We now define the problem formally and then discuss the key contributions.
\newcommand{\diversity}{\mathrm{div}}
\mparagraph{Problem Definition}
We are given a set $P$ of $n$ points in $\Re^d$ and a set of $m$ colors $C=\{c_1,\ldots, c_m\}$. Each point $p\in P$ is associated with a color $c(p)\in C$.
For any subset $S\subseteq P$, let $S(c_j)=\{p\in S\mid c(p)=c_j\}$ be the set of points from set $S$ with color $c_j$. We have $S(c_i)\cap S(c_j)=\emptyset$ for every pair $i<j$ and $\bigcup_{j}S(c_j)=S$.
%Finally, let $C(S)\subseteq C$ be the set of colors in $S$.
Let $\diversity(S)=\min_{p,q\in S}||p-q||_2$ be the diversity function representing the minimum pairwise Euclidean distance among points in $S$. For simplicity, throughout the paper we write $||p-q||$ to denote the Euclidean distance of points $p, q\in \Re^d$.
%We formally define the  problem we study.
\begin{definition}[\textbf{\fairDiv}]
 Given a set $P$ of $n$ points in $\Re^d$, where $d$ is a constant, a set $C$ of $m$ colors, and integers $k_1, k_2, \ldots, k_m$ such that $\sum_{j\in[1,m]}k_j=k$, the goal is to find a set $S^*\subseteq P$ such that $\diversity(S^*)$ is maximized, and for each $c_j\in C$, it holds that $|S^*(c_j)|\geq k_j$.
\end{definition}

Let $\gamma^*=\diversity(S^*)$ be the optimum diversity for the \fairDiv\ problem. For a parameter $\beta>1$, we say that an algorithm is a $\frac{1}{\beta}$-approximation for the \fairDiv\ problem if it returns a set $S$ with $\diversity(S)\geq \frac{1}{\beta}\gamma^*$, and $|S(c_j)|\geq k_j$, for every $c_j\in C$. The approximation ratio is $\frac{1}{\beta}$. Finally, we say that an algorithm is a $\frac{1}{\beta}$-approximation  with $(1-\eps)$-fairness if it returns a set $S$ with $\diversity(S)\geq \frac{1}{\beta}\gamma^*$, and $|S(c_j)|\geq \frac{k_j}{1-\eps}$, for every $c_j\in C$. 

\ansc{
Furthermore, we define the notion of \emph{coreset}, which is useful in the next sections. A set $G\subseteq P$ is a $(1+\eps)$-coreset for the \fairDiv\ problem if there exists a subset $\hat{S}\subseteq G$ such that $\hat{S}$ satisfies the fairness constraints and $\diversity(\hat{S})\geq \gamma^*/(1+\eps)$.}

%In this paper we do not only study the offline setting where we are given the full dataset upfront.

\ansb{
In many scenarios, the input consists of a stream of data which evolves over time, for example, tweets generated in real time or reviews of restaurants on Google or beer reviews on social media~\cite{beerdata}. In the context of Twitter, the goal is to choose a representative subset of tweets in real time originating from various geographic locations (diversity). This selection should ensure that every topic (politics, sports, etc.) is sufficiently represented by related tweets (fairness). In this setting, the above discussed methods would need to be run from scratch for each new tweet. Instead, 
we study the extension of \fairDiv\ problem when we receive the data in a \emph{streaming setting}. This problem has been studied in~\cite{wang2022streaming} with various applications in modern database systems.
In another example, an application might handle massive amounts of data that cannot be stored in memory to run an offline algorithm for the \fairDiv\ problem. Instead, a pass is made over the data storing and maintaining only a small subset of elements in memory (synopsis) that is used to get an approximate solution for the \fairDiv\ problem in the full dataset. %Instead, we make a pass over all items maintaining a small summary of them, and only after considering all items, compute a good enough solution for the \fairDiv\ problem using only the summary of data. 
During this pass, we maintain the synopsis efficiently under new insertions and, when needed, we should return a fair and diverse set representing all the items we have encountered in the stream. We have considered the beer reviews dataset to evaluate our techniques in a streaming setting (Section~\ref{sec:experiments}).
}

%The quality of the returned solution should be close to the fair and diverse set we would have returned if we had access to the entire input set.
%We also propose efficient algorithms in the \emph{streaming setting}, where we receive the input points in a streaming manner.

\begin{definition}[\textbf{\streamFairDiv}]
%We aim to solve the \fairDiv\ problem in the streaming setting.
We are given a set of colors $C$ and integers $k_1, k_2, \ldots, k_m$ such that $\sum_{j\in[1,m]}k_j=k$.
    At a time instance $t$, we receive a new point $p_t$ with color $c(p_t)\in C$.
    Let $P_t$ be the set of points we have received until time $t$. Over any time instance $t$, the goal is to maintain a ''small`` subset of points $\hat{P}_t\subseteq P_t$, such that, a solution for the \fairDiv\ problem in $P_t$ can be constructed efficiently using only the points stored in $\hat{P}_t$.
\end{definition}

A data analyst might want to run multiple queries exploring regions of data with fair and diverse representative sets. For example, someone might want to explore neighborhoods in Illinois that are both fair and diverse.
We study the \emph{range-query setting}, where the goal is to construct a data structure, such that given a query region, the goal is to return fair and diverse points in the query region in sub-linear time.
More formally, we define the next variation of the \fairDiv\ problem.
\begin{definition}
[\textbf{\queryFairDiv}]
    %We aim to solve the \fairDiv\ problem in the range query setting.
    Given a set $P$ of $n$ colors in $\Re^d$, and a set $C$ of $m$ colors, the goal is to construct a data structure, such that given a query rectangle $R$, 
    and integers $k_1, k_2, \ldots, k_m$ such that $\sum_{j\in[1,m]}k_j=k$, return a set $S^*\subseteq P\cap R$ such that $\diversity(S^*)$ is maximized, and for each $c_j\in C(S^*)$, it holds that $|S^*(c_j)|\geq k_j$. %The query time should run in $o(|P\cap R|)$ time.
\end{definition}

%\mparagraph{\fairClust}
%Given a set $P$, and integers $k_1, k_2, \ldots, k_m$ such that $\sum_{j\in[1,m]}k_j=k$, find a set $S^*\subseteq P$ such that i) $\max_{p\in P}\min_{s\in S^*} ||p-s||$ is minimized. and ii) for each $c_j\in C(S^*)$, $|S^*(c_j)|\geq k_j$.

\begin{table}[t]
\centering
%\resizebox{\linewidth}{!}{
 \begin{tabular}{c|c} 
 \hline
 Notation & Meaning\\\hline
 \hline
 $P$ & Point set\\ \hline
 $n$ & $|P|$\\ \hline
 $C$ & Set of colors\\\hline
 $m$ & $|C|$\\\hline
 $c(p)$ & Color of point $p\in P$\\\hline
 $S(c_j)$ & Set of points with color $c_j\in C$ in $S\subseteq P$\\\hline
 $k$ & Total output size (lower bound)\\\hline
 $k_j$ & Output size having color $c_j$ (lower bound)\\\hline
 $\eps$ & approximation error \\ \hline
 $S^*$ & Optimum solution for the \fairDiv\ problem\\\hline
 $\diversity(S)$ & Diversity of set $S$ (minimum pairwise distance)\\\hline
 $\gamma^*$ & $\diversity(S^*)$\\\hline
% $\gamma$ & Diversity parameter\\\hline
 $\mathcal{T}$& BBD-tree\\\hline
 $A$ & Matrix representation in MWU method\\\hline
 $h$ & probability vector in MWU method\\\hline
\end{tabular}
%}
\caption{Table of Notations}\label{Table:Notation}
\end{table}

\begin{table*}[t]
\resizebox{\linewidth}{!}{
\ansmeta{
\begin{tabular}{|c|c|c|c|c|c|c|}
\hline
\textbf{Problem}&\textbf{Method}&\multicolumn{2}{|c|}{\textbf{Time}}&\textbf{Space}&\textbf{Approximation}&\textbf{Fairness}\\\hline
\multirow{3}{8em}{\fairDiv / No Coreset}&\cite{addanki2022improved}&\multicolumn{2}{|c|}{$O(n^{\lambda})$}&$\Omega(n^2)$&$\frac{1}{2(1+\eps)}$&Exact, Randomized\\\cline{2-7}
&\cite{addanki2022improved}&\multicolumn{2}{|c|}{$O(nkm^3)$}&$O(n+km^2)$&$\frac{1}{(m+1)(1+\eps)}$&Exact, Deterministic\\\cline{2-7}
&\textbf{NEW}&\multicolumn{2}{|c|}{$O(nk)$}&$O(n)$&$\frac{1}{2(1+\eps)}$&$\frac{1}{1+\eps}$-approx., Randomized\\\hline\hline
\multirow{3}{8em}{\fairDiv / Coreset}&\cite{addanki2022improved}&\multicolumn{2}{|c|}{$O(nk+(mk)^{\lambda})$}&$\Omega(n+(km)^2)$&$\frac{1}{2(1+\eps)}$&Exact, Randomized\\\cline{2-7}
&\cite{addanki2022improved}&\multicolumn{2}{|c|}{$O(nk+k^2m^4)$}&$O(n+km^2)$&$\frac{1}{(m+1)(1+\eps)}$&Exact, Deterministic\\\cline{2-7}
&\textbf{NEW}&\multicolumn{2}{|c|}{$O(n+mk^2)$}&$O(n)$&$\frac{1}{2(1+\eps)}$&$\frac{1}{1+\eps}$-approx., Randomized\\\hline\hline
\multicolumn{2}{|c|}{}&\textbf{Update}&\textbf{Post-processing}&\multicolumn{3}{|c|}{}\\\hline
\multirow{3}{*}{\streamFairDiv}&\cite{wang2022streaming}&$O(\log\Delta)$&$O((mk)^2\log\Delta)$&$O(mk\log \Delta)$&$\frac{1-\eps}{3m+2}$&Exact, Deterministic\\\cline{2-7}
&\cite{addanki2022improved}&$O(\log\Delta)$&$O((mk)^\lambda\log\Delta)$&$O(mk\log \Delta)$&$\frac{1}{2(1+\eps)}$&Exact, Randomized\\\cline{2-7}
&\textbf{NEW}&$O(k)$&$O(mk^2)$&$O(mk)$&$\frac{1}{2(1+\eps)}$&$\frac{1}{1+\eps}$-approx., Randomized\\\hline\hline
\multicolumn{2}{|c|}{}&\textbf{Construction}&\textbf{Query}&\multicolumn{3}{|c|}{}\\\hline
\queryFairDiv&\textbf{NEW}&$O(n)$&$O(mk^2)$&$O(n)$&$\frac{1}{2(1+\eps)}$&$\frac{1}{1+\eps}$-approx., Randomized\\\hline\hline
\end{tabular}
}
}
\caption{\ansmeta{Comparison of our new algorithms with state--of--the--art.
%using coresets. In~\cite{}, we show the comparison of the algorithms without using coresets.
$\Delta=O(2^n)$ is the spread (max. over min. pairwise distance). $\lambda>2$ is the exponent such that an $\mathsf{LP}$ with $N$ constraints can be solved in $O(N^\lambda)$ time. For simplicity, we skip $\log^{O(1)}n$ factors.}}
\label{table:results}
\vspace{-2em}
\end{table*}

\subsection{Contributions}
\label{subsec:contr}
\ansmeta{In Table~\ref{table:results} we show our main results and compare them with the state--of--the--art methods.}

\ansmeta{
In Section~\ref{sec:algs}, we present \newAlg{}, the first near-linear time algorithm for the \fairDiv\ problem with constant approximation ratio. Our algorithm is also the first constant approximation algorithm that uses linear space. All previous constant approximation algorithms for \fairDiv\ have super-linear running time and super-linear space with respect to either $n$ or $k$.
%We map \fairDiv\ to a linear feasibility problem but instead of constructing and solving the linear program explicitly, we use the \emph{Multiplicative Weight Update} (MWU) method to find an approximation solution. %A straightforward implementation of the MWU approach would still require quadratic time and quadratic space.
%We combine the MWU approach with geometric data structures in a novel way, taking advantage of the geometry of the input set to solve the LP in near-linear time and linear space. Furthermore, we apply a rounding technique in near-linear time to derive the final solution.
For a constant $\eps\in(0,1)$, we get an $O(nk\log^3 n)$ time algorithm that returns a $\frac{1}{2(1+\eps)}$-approximation for the \fairDiv\ problem. %Our algorithm works for arbitrarily small values of $\eps$ with running time $O(\eps^{-d}kn)$. 
 The algorithm uses only $O(n)$ space. Each fairness constraint is satisfied approximately in expectation: If $S$ is the returned set then $\expect[S(c_j)]\geq \frac{k_j}{1-\eps}$, $\forall c_j\in C$.   }
    
If each $k_j\geq 3(1+\eps)\eps^{-2}\log(2m)$ is sufficiently large, in Section~\ref{sub:prob} we satisfy the fairness constraints approximately with probability at least $1-1/\delta$, in time $O(nk\log^3 n+n\log\frac{1}{\delta}\log n)$ and space $O(n)$. The approximation factor becomes $\frac{1}{6(1+\eps)}$.
%We note that the fairness and approximation factor are the same as in~\cite{addanki2022improved}, however our algorithm runs in near-linear time and linear space while the time and space in~\cite{addanki2022improved} is super-quadratic.

\ansc{In Section~\ref{sec:coreset} we show that any algorithm for the $k'$-center clustering can be used to derive a $(1+\eps)$-coreset of small size efficiently.
By constructing a coreset and then running our \newAlg{} algorithm we get a $\frac{1}{2(1+\eps)}$-approximation algorithm for the \fairDiv\ problem that runs in $O(n\log k + mk^2\log^3 k)$ time satisfying the fairness constraints approximately in expectation.
}

%A set $G\subseteq P$ is an $(1+\eps)$-coreset if there exists a subset $\hat{S}\subseteq G$ such that $\hat{S}$ satisfies all fairness constraints and $\diversity(\hat{S})\geq \gamma^*/(1+\eps)$. It is known how to find an $(1+\eps)$-coreset for the \fairDiv\ problem in \cite{addanki2022improved, wang2023max} however the coreset construction and the correctness proof heavily use the greedy Gonzalez~\cite{gonzalez1985clustering} algorithm for the $k$-center clustering problem. In Section~\ref{sec:coreset}, we show a more general result that allows us to construct a coreset faster. In fact, we show that any constant approximation algorithm for the $k'$-center algorithm, for $k'=O(\eps^{-2d}k)$ can be used to construct an $(1+\eps)$-coreset for the \fairDiv\ problem. Using our coreset as input to our algorithms, we get a $\frac{1}{2(1+\eps)}$-approximation algorithm for the \fairDiv\ problem that runs in $O(n\log k + mk^2\log^3 k)$ time, for constant $\eps$ satisfying the fairness constraints approximately in expectation or with high probability.
%Our result improves by a factor of $k$ the best previously known constant approximation algorithm using coresets that runs in $O(kn+(mk)^{2.37})$ time~\cite{addanki2022improved}.

\ansmeta{
The generality of our coreset construction allows us to extend our algorithms in different settings.
%For example, it is known how to maintain a constant approximation for the $k$-center clustering in the streaming model efficiently~\cite{agarwal2020efficient, guha2009tight, matthew2008streaming}.
In Section~\ref{sec:extensions}, we design an efficient streaming algorithm for the \streamFairDiv\ problem, called \newStreamAlg{}, maintaining a coreset for the \fairDiv\ problem.
Our new streaming algorithm
stores $O(mk)$ elements, takes $O(k\log k)$ update time per element for streaming processing, and $O(mk^2\log^3 k)$ time for post-processing to return a constant approximation for the \streamFairDiv\ problem satisfying the fairness constraints approximately.
%This is the first streaming algorithm for \fairDiv\ that does not depend on the spread (maximum over minimum pairwise distance) of the input set.
%\ansmeta{The quantities (space, update, post-processing time) of all previously known algorithms~\cite{wang2022streaming, addanki2022improved} for the \streamFairDiv\ problem depend on $\log \Delta$, where $\Delta$ is the spread of the input set.}
%There are two known algorithms for the \streamFairDiv\ problem. In~\cite{wang2022streaming}, they describe a streaming algorithm that stores $O(mk\log\Delta)$ elements in memory, takes $O(k\log\Delta)$ time per element for streaming processing, and $O(m^2k^2\log \Delta)=O(k^4\log\Delta)$ time for post-processing to return a $\frac{1-\eps}{3m+2}$-approximation for the \streamFairDiv\ problem. The parameter $\Delta$ is the spread of the input set. In~\cite{addanki2022improved}, they improved the approximation factor for the \streamFairDiv\ problem to a constant, however all asymptotic complexities still depend on $\log\Delta$.
In the range-query setting, %we can use the near-linear space data structure from~\cite{abrahamsen2017range} and our coreset construction to
we design a data structure of $O(n\log^{d-1}n)$ space in $O(n\log^{d-1} n)$ time, such that given a query rectangle $R$, a constant $\eps$, and parameters $k_1,\ldots, k_m$, it returns a set $S\subseteq P\cap R$ in $O(mk\log^{d-1} n + mk^2\log^3 k)$ time,
such that $S$ is a $\frac{1}{2(1+\eps)}$-approximation for the \fairDiv\ problem in $P\cap R$, and $\expect[S(c_j)]\geq \frac{k_j}{1+\eps}$, for every color $c_j\in C$.
%The result can be extended so that the fairness constraints are satisfied with high probability.
%This is the first efficient data structure for the \queryFairDiv\ problem.
    }
    
In Section~\ref{sec:experiments} we run experiments on real datasets showing that our new algorithms return diverse and fair results faster than the other baselines.
More specifically, among algorithms that return fair results with similar diversity, our algorithm is always the fastest one.
%In some cases, even if another baseline returns worse diversity, our algorithm is faster.
When another baseline is faster than our method it is always the case that the diversity of the set it returns is significantly worse than the diversity of the results returned by our algorithm.
Overall, \newAlg{} provides the best balance between diversity and running time.

\section{Preliminaries}
\label{sec:prelim}
\mparagraph{Known techniques for \fairDiv}
We first review the LP-based algorithm presented in~\cite{addanki2022improved} to find a solution for the \fairDiv\ problem.
They run a binary search over all possible pairwise distances. For a distance $\gamma$, they solve the following feasibility problem (LP).

\noindent\begin{minipage}[t][0ex][t]{0.7cm}
\lpproblem
\end{minipage}\vspace*{-2ex}
\begin{align}
\label{eq1}\sum_{p_i\in P(c_j)} x_i&\geq k_j \quad \quad \forall c_j\in C\\
\label{eq2} \sum_{p_i\in P\cap \mathcal{B}(p,\gamma/2)} x_i &\leq 1  \quad \quad \forall p\in P\\
\label{eq3}1\geq x_{i} &\geq 0, \quad \quad \forall p_i\in P
\end{align}
$\mathcal{B}(p,\gamma/2)$ represents a ball with center $p$ and radius $\gamma/2$.
If ($\mathsf{LP}$\ref{lp:1}) is infeasible, they try smaller values of $\gamma$. Otherwise, they try larger values of $\gamma$. Then they describe a rounding technique to construct a valid solution for the \fairDiv\ problem.
Let $\hat{x}$ be the solution of ($\mathsf{LP}$\ref{lp:1}) corresponding to largest $\hat{\gamma}$ that the LP was feasible.
They generate a random ordering $\sigma$ of $[n]$ as follows: $\sigma(t)$ is randomly chosen from $R_t=[n]\setminus\{\sigma(1),\ldots, \sigma(t-1)\}$ such that a number $b\in R_t$ is chosen with probability $\Pr[\sigma(t)=b]=\frac{\hat{x}_b}{\sum_{\ell\in R_t}\hat{x}_\ell}$. After generating the ordering $\sigma$, they construct the output set $S\subseteq P$ including the point $p_j\in S$ if and only if $\sigma(j)\leq\sigma(\ell)$ for all $p_\ell\in P\cap \mathcal{B}(p_j,\gamma/2)$.
They show the following theorem.
\begin{theorem}[\cite{addanki2022improved}]
\label{thm:LP}
    The algorithm described above returns a set $S$ with $\diversity(S)\geq \gamma^*/2$ such that for each color $c_j$, $\expect[|S(c_j)|]\geq k_j$.
\end{theorem}
%Indeed, this LP-based algorithm returns a constant approximation for the \fairDiv\ problem. 
Notice that both the space and the running time of the algorithm is $\Omega(n^2)$. Specifically, it needs $O(n^2)$ space only to represent ($\mathsf{LP}$\ref{lp:1})
because $|P\cap \mathcal{B}(p,\gamma/2)|=O(n)$ for every point $p\in P$.
%ii) All steps of the LP-based algorithm take $\Omega(n^2)$ time. Identifying all possible pairwise distances in $P$ to run a binary search on the distances takes $O(n^2)$ time. Constructing each LP takes $O(n^2)$ time.
Solving each instance of ($\mathsf{LP}$\ref{lp:1}) takes $O(n^{\lambda})$ time, where $\lambda$ is the exponent such that an \textsf{LP} with $N$ variables and $N$ constraints can be solved in $O(N^\lambda)$ time\footnote{In any case, $\lambda\geq 2$. Currently, the best theoretical algorithm solving an \textsf{LP} has $\lambda\geq 2.37$~\cite{jiang2021faster}.
%More practical polynomial time algorithms have $\lambda\geq 3.5$~\cite{}.
}.
The rounding algorithm is executed in $O(n^2)$ time. Overall, the running time is $O(n^{\lambda}\log n)$.

\ansmeta{
In the same paper they also propose the Fair-Greedy-Flow algorithm that returns a $\frac{1}{(m+1)(1+\eps)}$-approximation in $O(nkm^3\log n)$ time, for constant $\eps$. The algorithm maps the \fairDiv\ problem to a max-flow instance with $O(km)$ nodes and $O(mk^2)$ edges.
}

The authors in~\cite{addanki2022improved} also described a $(1+\eps)$-coreset for the \fairDiv\ problem.
%In particular they construct a set $G\subseteq P$ of size $O(\eps^{-d}km)$.
%such that, any $\frac{1}{\beta}$-approximation algorithm for the \fairDiv\ problem on $G$, returns a $\frac{1}{(1+\eps)\beta}$-approximation on $P$.
They run the well known Gonzalez algorithm~\cite{gonzalez1985clustering} for the $k'$-center clustering problem in each set $P(c_j)$ independently, for $k'=\eps^{-d}k$. Let $G_j$ be the solution of the Gonzalez's algorithm in $P(c_j)$. Then, $G=\bigcup_{c_j\in C}G_j$. It holds that $|G|=O(\eps^{-d}km)$ and it is constructed in $O(\eps^{-d}nk)$ time.
If we combine the results in Theorem~\ref{thm:LP} with the coreset construction for a constant $\eps$, we get, an algorithm that returns a set $S$ with $\diversity(S)\geq \frac{1}{2(1+\eps)}\gamma^*$ in $O(kn+(km)^{\lambda})$ time such that for each color $c_j$, $\expect[|S(c_j)|]\geq k_j$.
\ansmeta{To satisfy fairness exactly, we can combine the coreset with the Fair-Greedy-flow algorithm~\cite{addanki2022improved} to get a $\frac{1}{(m+1)(1+\eps)}$-approximation in $O(kn+k^2m^4\log k)$ time.}

\paragraph{Diversity with high probability}
All the results above, return a set $S$ that satisfies fairness in expectation. The authors in~\cite{addanki2022improved} extended the results to hold with probability at least $1-1/n$, also called with high probability.
Given $\hat{x}$, the solution from ($\mathsf{LP}$\ref{lp:1}), they convert it to a solution $\hat{y}$ for the following (non-linear) feasibility problem.

\noindent\begin{minipage}[t][0ex][t]{0.7cm}
\fpproblem
\end{minipage}\vspace*{-2ex}
\begin{align}
\label{eq2.1}\sum_{p_i\in P(c_j)} y_i&\geq k_j,  \quad\quad \forall c_j\in C\\
\label{eq2.2} \sum_{p_i\in P\cap\mathcal{B}(p,\gamma/6)} y_i & \leq 1,  \quad \quad \forall p\in P\\
\label{eq2.3}y_{i} &\geq 0, \quad \quad \forall p_i\in P\\
y_i&>0\text{ and }y_\ell>0 \Rightarrow ||p_i-p_{\ell}||\geq\frac{\gamma}{3}, \\
\nonumber&\quad\quad\quad\quad\forall p_i,p_\ell\in P(c_j), \forall c_j\in C
\end{align}
If $k_j\geq 3\eps^{-2}\log(2m)$ for every $c_j\in C$, the authors showed that, if they apply the same rounding technique as in the the previous case they return a set $S$ such that $\diversity(S)\geq \gamma^*/6$, where $|S(c_j)|\geq \frac{k_j}{1-\eps}$ for every $c_j\in C$, with probability at least $1-1/n$. Unfortunately, they still need to solve ($\mathsf{LP}$\ref{lp:1}) to derive this result. Even without solving ($\mathsf{LP}$\ref{lp:1}), the algorithm they propose to convert the solution from ($\mathsf{LP}$\ref{lp:1}) $\hat{x}$ to a solution for ($\mathsf{FP}$\ref{fp:1}) $\hat{y}$ takes $\Omega(n^2)$ time. Hence, the overall running time is super-quadratic. If the coreset is used, then we get a $\frac{1}{6(1+\eps)}$-approximation algorithm in $O(kn+(mk)^{2.37})$ time satisfying the fairness constraints with high probability.

\mparagraph{Geometric data structures}
We describe the main geometric data structure we use in the next sections.
\paragraph{BBD-tree} The main geometric data structure we use is the \emph{BBD-tree}~\cite{arya2000approximate, arya1998optimal}, which is a variant of the quadtree~\cite{finkel1974quad}.
A BBD-tree $\mathcal{T}$ on a set $P$ of n points in $\Re^d$ is a binary
tree of height $O(\log n)$ with exactly $n$ leaves. Let $\square$ be the smallest
axis-aligned hypercube containing $P$. Each node $u$ of $\mathcal{T}$ is associated
with a region $\square_u$, which is either a rectangle or a region between
two nested rectangles, and a subset $P_u\subseteq P$ of points that lie inside
$\square_u$. Notice that $\square_{root}=\square$. If $|P_u|=1$, then $u$ is a leaf. If $|P_u| > 1$,
then $u$ has two children, say, $w$ and $z$, and $\square_w$ and $\square_z$ partition
$\square_u$. Regions associated with the nodes of $\mathcal{T}$ induce a
hierarchical partition of $\Re^d$. A BBD tree has $O(n)$ space and can be constructed in $O(n\log n)$ time. Given a parameter $\eps\in (0,1)$ and a ball $\mathcal{B}(x, r)$ in $\Re^d$, the BBD-tree runs the query procedure $\mathcal{T}(x,r)$ that
returns a set of nodes $\mathcal{U}(x,r)=\{u_1,\ldots, u_\kappa\}$ from $\mathcal{T}$ (also called \emph{canonical nodes}) for $\kappa=O(\log n + \eps^{-d})$ in $O(\log n+\eps^{-d})$ 
time such that $\square_{u_i}\cap \square_{u_j}=\emptyset$ for every pair $1\leq i<j\leq \kappa$, and
$\mathcal{B}(x,r)\subseteq \bigcup_{1\leq i\leq \kappa}\square_{u_i}\subseteq \mathcal{B}(x,(1+\eps)r)$.
%In other words, the BBD tree returns a set of rectangles (with or without holes) such that the union of them completely cover the ball $\mathcal{B}(x,r)$ and might cover some part of $\mathcal{B}(x,(1+\eps)r)\setminus \mathcal{B}(x,r)$.
By reporting all points $P_{u_i}$ for $i\leq \kappa$, the BBD tree can be used for reporting all points in $P\cap \mathcal{B}(x,r)$ along with some points from $P\cap (\mathcal{B}(x,(1+\eps)r)\setminus \mathcal{B}(x,r))$.

%The BBD-tree can also be used for $\eps$-approximate Nearest Neighbor queries. Given a set $P\in \Re^d$ of $n$ points, a BBD-tree can be constructed, such that, given a query point $q$, a point $p$ is returned by the BBD-tree such that $||p-q||\leq (1+\eps)\min_{t\in P}||p-t||$, in $O(\log n +\eps^{-d})$ time.
\paragraph{WSPD}
Using a quadtree~\cite{finkel1974quad}, someone can get 
a \emph{Well Separated Pair Decomposition} (WSPD)~\cite{callahan1995decomposition, har2005fast} in $P\in \Re^d$. In $O(\eps^{-d}n\log n)$ time, we can construct a list $\mathcal{L}=\{L_1, \ldots, L_z\}$ of $z=O(\eps^{-d}n)$ distances, such that for every pair $p,q\in P$, there exists $L_j\in \mathcal{L}$ such that $(1-\eps)||p-q||\leq L_j\leq (1+\eps)||p-q||$.

\mparagraph{Multiplicative Weight Update (MWU) method}
The MWU method is used to solve the following linear feasibility problem.
\begin{equation}
    \label{feasibility}
    \exists x\in \mathcal{P}: Ax\leq b,
\end{equation}
where $A\in \Re^{m'\times n'}$, $x\in \Re^{n'}$, $b\in \Re^{m'}$, $Ax\geq 0$, $b\geq 0$, and $\mathcal{P}$ is a convex set in $\Re^{n'}$. Intuitively, $\mathcal{P}$ captures the ``easy'' constraints to satisfy while $A$ represents the ``hard'' constraints to satisfy.
The authors in~\cite{arora2012multiplicative} describe an iterative algorithm using a simple \textsf{ORACLE}.
Let \textsf{ORACLE} be a black-box procedure that solves the following single linear constraint for a probability vector $h\in \Re^{m'}$.
\begin{equation}
\label{eq:oracle}
\exists x\in \mathcal{P}: h^\top Ax\leq h^\top b.    
\end{equation}

The \textsf{ORACLE} decides if there exists an $x$ that satisfies the single linear constraint. Otherwise, it returns that there is no feasible solution.
A $\rho$-\textsf{ORACLE} is an \textsf{ORACLE} such that whenever \textsf{ORACLE} manages to find a feasible solution $\hat{x}$ to problem~\eqref{eq:oracle}, then $A_i\hat{x}-b_i\in [-1,\rho]$ for each constraint $i\in[m']$, where $A_i$ is the $i$-th row of $A$.
%\begin{align*}
%\forall i\in I &\quad A_ix\geq b_i \in [-\lambda,\rho]\\
%\forall i\notin I &\quad  A_ix\geq b_i \in [-%\rho,\lambda]
%\end{align*}

The algorithm starts by initializing $h$ to a uniform probability vector with value $1/m'$. In each iteration the algorithm solves Equation~\eqref{eq:oracle}. If \eqref{eq:oracle} is infeasible, we return that the original feasibility problem in Equation~\eqref{feasibility} is infeasible. Let $x^{(t)}$ be the solution of the problem in Equation~\eqref{eq:oracle} in the $t$-th iteration of the algorithm.
Let $\delta_i=\frac{1}{\rho}(A_ix^{(t)}-b_i)$. We update $h[i]=(1+\delta_i\cdot\eps/4)h[i]$, where $h[i]$ is the $i$-th element of vector $h$. We continue in the next iteration defining a new feasibility problem with respect to the new probability vector $h$. After $T=O(\rho\log(m')/\eps^2)$ iterations, if every oracle was feasible, they return $x^*=\frac{1}{T}\sum_{t=1}^Tx^{(t)}$.
Otherwise, if an oracle was infeasible, they argue that the initial problem is infeasible.
\ansmeta{Overall, every algorithm using the MWU method to solve a problem in the form of Equation~\eqref{feasibility} should implement two procedures: $\mathsf{Oracle}(\cdot)$ that implements a $\rho$-\textsf{ORACLE} and $\mathsf{Update}(\cdot)$ that updates the probability vector $h$.}
In~\cite{arora2012multiplicative} they prove the following theorem.
%\vspace{-2em}
\begin{theorem}[\cite{arora2012multiplicative}]
\label{thm:mutli-weights}
\ansmeta{
Given a feasibility problem as defined above, a parameter $\eps$, a $\rho$-\textsf{ORACLE} implemented in procedure $\mathsf{Oracle}(\cdot)$, and an update procedure $\mathsf{Update}(\cdot)$, there is an algorithm which either finds an $x$ such that $\forall i$, $A_ix_i\leq b_i+\eps$ or correctly concludes that the system is infeasible. The algorithm makes $O(\rho\log(m')/\eps^2)$ calls to procedures $\mathsf{Oracle}(\cdot)$ and $\mathsf{Update}(\cdot)$.}
\end{theorem}

\newcommand{\LPone}{($\mathsf{LP}$\ref{lp:1})}
\newcommand{\LPtwo}{($\mathsf{LP}$\ref{lp:2})}
\section{Efficient algorithm for \fairDiv}
\label{sec:algs}
In this section we propose two efficient algorithms for the \fairDiv\ problem. The first one guarantees approximate fairness in expectation, while the second one guarantees approximate fairness with high probability.

\subsection{Diversity in expectation}
\label{sub:exp}

\begin{figure*}
    \includegraphics[width=.9\textwidth]{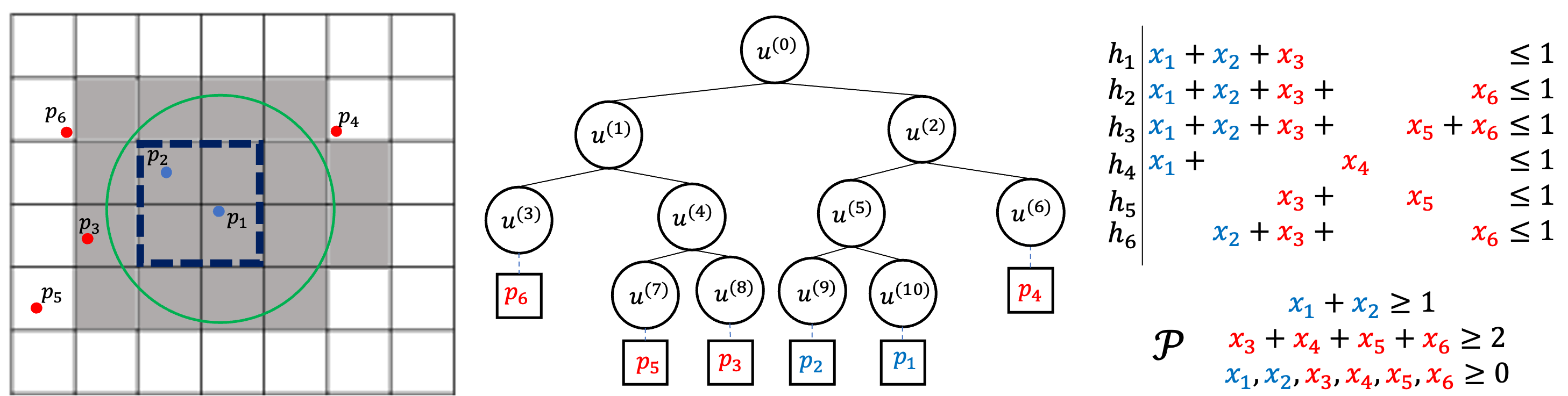}
    \caption{Left: Input set of points. Middle: Simplified BBD tree, Right: The decision problem $\exists x\in \mathcal{P}$ s.t. $h^\top Ax\leq b$.\label{fig:grid}}
\end{figure*}

%\begin{example}
%    We
%\end{example}

%\subsection{Diversity in expectation}
%\label{sub:exp}

\ansmeta{
\mparagraph{High-level idea}
Recall the LP-based algorithm proposed in~\cite{addanki2022improved}: solve \LPone\ using an LP solver and then round the solution as described in Section~\ref{sec:prelim}. We design a new algorithm that uses the MWU approach to solve a modified feasibility problem.
While the MWU approach can work directly on \LPone\, it takes $\Omega(n^2)$ time to run. Instead, we define a new linear feasibility problem, called \LPtwo\ and we use the MWU method to approximately solve \LPtwo\ in near-linear time. Finally, we round its fractional solution to return a valid solution for the \fairDiv\ problem in near-linear time using advanced geometric data structures.

Assume that $\gamma$ is a pairwise distance among the items in $P$. Our algorithm checks whether there exists a set $S\subseteq P$ that satisfies the fairness constraints such that $\diversity(S)\geq \frac{\gamma}{2(1+\eps)}$. We map this decision problem to a new linear feasibility problem \LPtwo. The Constraints~\eqref{eq1} and Constraints~\eqref{eq3} from \LPone\ remain the same.
However, we slightly modify Constraints~\eqref{eq2}.
For a point $p$ we define a set $S_p^\eps\subseteq P$ denoting its ``neighboring'' points, with a definition of neighboring which is convenient for the data structure we use. The set $S_p^\eps$ contains all points within distance $\frac{\gamma}{2(1+\eps)}$ from $p$, might contain some points within distance $\gamma/2$, and no point with distance more than $\gamma/2$.
%It does not contain any point with distance more than $\gamma/2$ from $p$.
The properties of the BBD tree are used to formally define $S_p^{\eps}$.
We define $S_p^{\eps}=\{p\in \square_{u_i}\cap P\mid u_i\in \mathcal{U}(p,\frac{\gamma}{2(1+\eps)})\}$, i.e., the set of points in the canonical nodes returned by query $\mathcal{T}(p,\frac{\gamma}{2(1+\eps)})$.
We replace Constraints~\eqref{eq2} with
$\sum_{p_i\in S_p^\eps} x_i\leq 1, \forall p\in P$.
Overall, the new feasibility problem is:

\noindent\begin{minipage}[t][0ex][t]{0.7cm}
\lpproblem
\end{minipage}\vspace*{-2ex}
\begin{align}
\label{modeq1}\sum_{p_i\in P(c_j)} x_i&\geq k_j \quad \quad \forall c_j\in C\\
\label{eq:new} \sum_{p_i\in S_p^\eps} x_i &\leq 1  \quad \quad \forall p\in P\\
\label{modeq3}1\geq x_{i} &\geq 0, \quad \quad \forall p_i\in P
\end{align}

%\begin{equation}
%\label{eq:new}
%\sum_{p_i\in S_p^\eps} x_i\leq 1, \quad\quad \forall p\in P.
%\end{equation}

%The matrix $A$ consists of the factors in the group of constraints $\sum_{p_i\in\mathcal{B}(p,\gamma/2)}x_i\leq 1$.

We use the MWU method to compute a feasible solution for \LPtwo. Recall that the MWU method solves feasibility problems in the form of Equation~\eqref{feasibility}, $\exists x\in \mathcal{P}:Ax\leq b$. Next, we show that \LPtwo\ can be written in this form by defining $\mathcal{P}$, $A$, and $b$.

Instead of considering that the trivial constraints $\mathcal{P}$ contains  only the inequalities $1\geq x_i\geq 0$, we assume that Constraints~\eqref{modeq1} are also trivial and contained in $\mathcal{P}$. This will allow us later to design a $k$-\textsf{ORACLE}. The set $\mathcal{P}$ is convex because it is defined as the intersection of $m+n$ halfspaces in $\Re^n$. Hence, it is valid to use the MWU method.
The new Constraints~\eqref{eq:new} define the binary square matrix $A$, having one row for every point $p\in P$. The value $A[\ell,i]=1$ if $p_i\in S_{p_\ell}^\eps$, otherwise it is $0$. 
Finally, $b$ is defined as a vector in $\Re^d$ with all elements being $1$. From Theorem~\ref{thm:mutli-weights}, we know that the MWU method returns a solution that satisfies the constraints in $\mathcal{P}$ (Constraints~\eqref{modeq1} and~\eqref{modeq3}) exactly, while the Constraints~\eqref{eq:new} are satisfied with an $\eps$ additive error.

A straightforward implementation of the MWU method over \LPtwo\ would still take super-quadratic time to run; even the computation of $A$ takes $O(n^2)$ time. Our new algorithm does not construct the matrix $A$ explicitly. Likewise, it does not construct the sets $S_p^\eps$ explicitly. Instead, we use geometric data structures to implicitly represent $A$ and $S_p^\eps$ and execute our algorithm in near-linear time.
}

Before we continue with our method, we note that there has been a lot of work related to implicitly solving special classes of LPs in computational geometry using the MWU method~\cite{chan2020faster, chekuri2020fast, clarkson2005improved}. Despite the geometric nature of our problem, to the best of our knowledge, all previous (geometric) techniques that implicitly solve an $\mathsf{LP}$ do not extend to our problem. For example, in~\cite{chekuri2020fast} they describe algorithms for geometric problems using the MWU method with running time that depends linearly on the number of non-zero numbers in matrix $A$. In our problem definition there might be $O(n^2)$ non-zero numbers in $A$, hence these algorithms cannot be used to derive near-linear time algorithms for our problem. In ~\cite{chan2020faster, chekuri2020fast} they also provide near-linear time constant approximation algorithms for some geometric problems such as the geometric set cover, covering points by disks, or the more relevant to our problem, maximum independent set of disks. However there are three major differences with our problem definition: i) These algorithms work only in $\Re^2$ or $\Re^3$. We propose approximation algorithms for any constant dimension $d$. ii) They do not consider fairness constraints and it is not clear how to extend their methods to satisfy fairness constraints. iii) The optimization problem in \fairDiv\ is different from the optimization problems they study.

\ansmeta{
\mparagraph{Example}
%\eparagraph{Example}
We use a toy example to describe our new algorithm in the next paragraphs. In Figure~\ref{fig:grid} (Left) we show the input points that consists of two blue points $p_1, p_2$ and four red points $p_3, p_4, p_5, p_6$. Let blue color be $c_1$ and red color be $c_2$. The goal is to solve the \fairDiv\ problem among the five points with $k_1=1$ and $k_2=2$, i.e., our solution should contain at least one blue and at least two red points.
Throughout the example, we assume that $\gamma=5$ and $\eps=1$.
In Figure~\ref{fig:grid} (Middle) we show a simplified version of the BBD tree constructed on the input set. Each node corresponds to a rectangle in $\Re^d$. For example, the node $u^{(5)}$ corresponds to the dashed rectangle that contains $p_1$ and $p_2$ in Figure~\ref{fig:grid} (Left). The leaf nodes in the BBD tree correspond to the non-empty cells in Figure~\ref{fig:grid} (Left). For simplicity, we assume that each grid cell has diagonal $1$ (which is equal to the maximum error $\eps$). Consider the blue point $p_1$. The green circle with center $p_1$ has radius $\gamma/(2(1+\eps))=1.25$. By definition, all points that lie in grid cells that are intersected by the green circle (the gray cells in Figure~\ref{fig:grid} (Left)) belong in $S_{p_1}^{\eps}$, i.e. $S_{p_1}^{\eps}=\{p_1,p_2, p_3\}$. Notice that $p_3\in S_{p_1}^\eps$ because the green circle around $p_1$ intersects the grid cell (associated with the node $u^{(8)}$) that $p_3$ belongs to. On the other hand $p_4\notin S_{p_1}^\eps$. Interestingly, $S_{p_4}^\eps=\{p_1, p_4\}$ because a ball of radius $1.25$ and center $p_4$ intersects the cell that $p_1$ belongs to. %Hence, it is possible, that $q\notin S_{p}^\eps$, while $p\in S_{q}^\eps$, for two points $p,q$.
We also have $S_{p_2}^\eps=\{p_1, p_2, p_3, p_6\}$, $S_{p_3}^\eps=\{p_1, p_2, p_3, p_5, p_6\}$, $S_{p_5}^\eps=\{p_3, p_5\}$, and $S_{p_6}^\eps=\{p_2, p_3, p_6\}$
Using the sets $S_p^\eps$, in Figure~\ref{fig:grid} (Right), we show the ($\mathsf{LP}$\ref{lp:2}) for our instance. The trivial constraints $\mathcal{P}$ are shown in the bottom, while the main constraints are shown at the top.

%\end{example}
%Hence, our feasibility problem has exactly the form of the linear feasibility problem~\eqref{feasibility}.
%As shown in~\eqref{eq:oracle}, in each iteration of the MWU method we solve the feasibility problem $h^\top Ax\leq 1\Leftrightarrow (h_1+h_2+h_3+h_4)\cdot x_1+(h_1+h_2+h_3+h_6)\cdot x_2+(h_1+h_2+h_3+h_5+h_6)\cdot x_3+(h_4)\cdot x_4+(h_3+h_5)\cdot x_5+(h_2+h_3+h_6)\cdot x_6\leq 1$, given that $x\in \mathcal{P}$.
%}

\mparagraph{New Algorithm}
%We use the MWU approach to design an algorithm for the \fairDiv\ problem.
Our new algorithm is called \textbf{M}ultiplicative weight update method for \textbf{F}air \textbf{D}iversification (\newAlg{}) and it is described in Algorithm~\ref{alg:alg0}. First it computes a sorted array $\Gamma$ of (candidates of) pairwise distances in $P$. Then it runs a binary search on $\Gamma$. Lines 2, 4, 5, 13, 14, 16 are all trivial executing the binary search on $\Gamma$. Each time we find a feasible (infeasible) solution for our optimization problem ($\mathsf{LP}$\ref{lp:2}) we try larger (smaller) values of $\gamma$. For a particular $\gamma\in \Gamma$, in lines 4--16 we use the MWU method to solve ($\mathsf{LP}$\ref{lp:2}). The algorithm follows the MWU method as described in the end of Section~\ref{sec:prelim}. In particular, for at most $T=O(\eps^{-2}\rho\log n)$ iterations, it calls $\mathsf{Oracle}(\cdot)$ in line 9 to decide whether there exists $x$ such that $h^\top Ax\leq h^\top b$ and $x\in \mathcal{P}$. Recall that in our case $b\in\Re^n$ and $b=\{1,\ldots, 1\}$, so it is sufficient to decide whether there exists $x$ such that $h^\top Ax\leq \sum_{\ell}h[\ell]\Leftrightarrow h^\top Ax\leq 1$ and $x\in \mathcal{P}$. If $\mathsf{Oracle}()$ returns a feasible solution $\bar{x}$, in line 11 it updates the final solution $\hat{x}$ and uses the $\mathsf{Update}(\cdot)$ procedure to update the vector $h$ based on $\bar{x}$. If all $T$ iterations return feasible solutions, in line 15 it computes the final solution of ($\mathsf{LP}$\ref{lp:2}) for a given $\gamma$.
In the end, in line 17 we run a rounding procedure to derive the final set of points $S$.
\begin{algorithm}[t]
\ansmeta{
    \caption{\newAlg$(P,\eps, k_1,\ldots, k_m)$}
    \label{alg:alg0}
    $\Gamma\gets$ Sorted array of pairwise distances in $P$\;
    $M_l\gets 0$,\hspace{0.25em} $M_u\gets |\Gamma|-1$,\hspace{0.25em} $k\gets k_1+\ldots+k_m$\;
    $T\gets O(\eps^{-2}\rho\log n)$\;
    \While{$M_l\neq M_u$}{
        $M\gets\lceil (M_l+M_u)/2 \rceil$, \hspace{0.25em}
        $\gamma\gets \Gamma[M]$\;
         $\hat{x}\gets(0,\ldots,0)^\top\in \Re^n$\;
        $h\gets (\frac{1}{n}, \ldots, \frac{1}{n})^\top\in \Re^n$\;
        \For{$1,\ldots, T$}{
            $\bar{x}\gets$\textsf{Oracle}$(P,h,\gamma,\eps, k_1,\ldots, k_m)$\;
            \If{$\bar{x}\neq\emptyset$}{
                $\hat{x}\gets \hat{x}+\bar{x}$\;
                $h\gets$\textsf{Update}$(P, \bar{x},\gamma,\eps)$\;
            }\Else{
                $M_u\gets M-1$,
                Go to Line 4\;
            }
        }
        $\hat{x}\gets \hat{x}/T$\;
        $M_l\gets M$\;
    }
    $S\gets$\textsf{Round}$(P, \eps, \hat{x})$\;
    \Return $S$\;
}
\end{algorithm}
Overall, the algorithm follows the high level idea of the LP-based algorithm in~\cite{addanki2022improved}, using the MWU method~\cite{arora2012multiplicative} instead of an LP solver.
If $Q_{\Gamma}, Q_O, Q_U, Q_R$ is the running time to compute $\Gamma$, and run $\mathsf{Oracle}(\cdot)$ $\mathsf{Update}(\cdot)$, and $\mathsf{Round}(\cdot)$, respectively, then the overall running time of Algorithm~\ref{alg:alg0} is
$O\left(Q_{\Gamma} + (\eps^{-2}\rho \log n)(Q_O+Q_U+n)\log |\Gamma| + Q_R\right)$.
%Without any further action, Algorithm~\ref{alg:alg0} still takes $\Omega(n^2)$ time to run: the set $\Gamma$ takes $\Omega(n^2)$ time to compute, both the procedures $\mathsf{Oracle}$ and $\mathsf{Update}$ take $\Omega(n^2)$ to run as described in~\cite{arora2012multiplicative}, and the procedure $\mathsf{Round}$ takes $\Omega(n^2)$ to run as described in~\cite{addanki2022improved}. Furthermore, the procedure $\mathsf{Oracle}$ uses $O(n^2)$ space to run as described in~\cite{arora2012multiplicative}.
Using the results in~\cite{arora2012multiplicative} and~\cite{addanki2022improved}, $Q_{\Gamma}, Q_O, Q_U, Q_R = \Omega(n^2)$, leading to a super-quadratic time algorithm.
In the next paragraphs, we use geometric tools to show how we can find a set $\Gamma$ and run all the procedures $\mathsf{Oracle}$, $\mathsf{Update}$, $\mathsf{Round}$ only in near-linear time using only linear space with respect to $n$. We also show that $\rho= k$.
}

%As we will see later the grid represent the smallest cells of the BBD tree. Without loss of generality, we assume that each cell in Figure~\ref{fig:grid} has diagonal $1$.

%}

\mparagraph{The $\mathsf{Oracle}(\cdot)$ procedure}
\begin{algorithm}[t]
\ansmeta{
    \caption{\textsf{Oracle}$(P, h, \gamma, \eps,k_1,\ldots, k_m)$}
    \label{alg:alg1}
    $\mathcal{T}\gets \text{ BBD tree on } P$\;
    \lForEach{$u\in\mathcal{T}$}{
        $u_s\gets 0$
    }
   
    \ForEach{$p_\ell\in P$}{
        $\mathcal{U}_{p_\ell}\gets \mathcal{T}(p_\ell, \frac{\gamma}{2(1+\eps)})$\;
        \lForEach{$u\in \mathcal{U}_{p_\ell}$}{
            $u_s\gets u_s+h[\ell]$
        }
    }
    \ForEach{$p_i\in P$}{
        $w_i\gets 0$\;
        $v\gets \text{ leaf node of }\mathcal{T} \text{ such that } p_i\in \square_v\cap P$\;
        \ForEach{$u$ in the path from $v$ to the root of $\mathcal{T}$}{
            $w_i\gets w_i+u_s$
        }
    }
    $\bar{x}=(0,\ldots, 0)\in \Re^n$\;
    \ForEach{$c_j\in C$}{
        $W_j\gets k_j$-th smallest weight in $\{w_i\mid p_i\in P(c_j)\}$\;
        $P_j\gets \{p_i\in P(c_j)\mid w_i\leq W_j\}$\;
        \lForEach{$p_i\in P_j$}{
            $\bar{x}_i=1$
        }
    }
    \lIf{$\sum_{p_i\in P}\bar{x}_iw_i\leq 1$}{
        \Return $\bar{x}$
    }\lElse{
        \Return $\emptyset$ (\textsf{Infeasible})
    }
}
\end{algorithm}
\ansmeta{We design a $k$-\textsf{ORACLE} procedure as defined in Section~\ref{sec:prelim}. The goal is to decide whether there exists $x$ such that $h^\top Ax\leq 1$ and $x\in \mathcal{P}$. We note that $\mathsf{Oracle}(\cdot)$ does not compute the matrix $A$ explicitly.}
In fact, $h^\top Ax$ can be written as $\sum_{p_i\in P}\alpha_ix_i$, for some real coefficients $\alpha_i$.
Intuitively, for every color $c_j$, our goal is to find the $k_j$ points with color $c_j$, having the smallest coefficients $\alpha_i$. Our algorithm first finds all coefficients $\alpha_i$ and then it chooses the $k_j$ smallest from each color $c_j$. In that way, we find the solution $\bar{x}$ that minimizes $h^\top A\bar{x}$ for $\bar{x}\in\mathcal{P}$. Finally, we check whether $h^\top A\bar{x}\leq 1$.

We are given a probability vector $h\in \Re^n$. Each value $h[i]\in h$ corresponds to the \emph{weight} of the $i$-th row of matrix $A$. In other words each point $p_\ell\in P$ is associated with a weight $h[\ell]$. We build a slightly modified BBD tree $\mathcal{T}$ over the weighted points $P$.
%to handle weighted SUM queries.
Let $\mathcal{T}$ be the tree constructed as described in Section~\ref{sec:prelim} over the set of points $P$.
For every node $u$ of $\mathcal{T}$, we initialize a weight $u_s=0$.
%For a point $p\in P$ and a radius $r$, let $\mathcal{T}(p,r)$ be a range query returning  $\kappa=O(\log n + \eps^{-d})$ canonical nodes $\mathcal{U}(p,r)=\{u_1,\ldots, u_\kappa\}$ of $\mathcal{T}$ such that, i) for a point $q\in P$, if $||p-q||\leq r$, then $\exists u_i\in \mathcal{U}(p,r)$ such that $q\in \square_{u_i}\cap P$, and ii) if $q\in \square_{u_i}\cap P$ then $||p-q||\leq (1+\eps)r$.
The data structure $\mathcal{T}$ has $O(n)$ space and can be constructed in $O(n\log n)$ time.

%We first define $S_p^\eps=\emptyset$. For each $p\in P$ we define $\mathcal{C}_p$ as described above. For each $q\in \mathcal{C}_p$ we set $S_p^\eps\leftarrow S_p^\eps\cup\{p\}$. After processing all points $p\in P$, $S_p^\eps$ contains all points in $\mathcal{B}(p,r)$ and might contain some points in $\mathcal{B}(p,(1+\eps)r)\setminus \mathcal{B}(p,r)$.

Using our modified BBD tree $\mathcal{T}$, in Algorithm~\ref{alg:alg1} we show how to check whether there exists $\bar{x}$ such that $h^\top A\bar{x}\leq 1$ and $\bar{x}\in \mathcal{P}$.
For each $p_\ell\in P$ we run the query $\mathcal{T}(p_\ell, \frac{\gamma}{2(1+\eps)})$ and we get the set of canonical nodes $\mathcal{U}_{p_\ell}:=\mathcal{U}(p_\ell,\frac{\gamma}{2(1+\eps)})$. For each node $u\in\mathcal{U}_{p_\ell}$, we update $u_s\leftarrow u_s+h[\ell]$.
After we consider all points, 
we revisit each point $p_i\in P$ and continue as follows: We initialize a weight $w_i=0$. We start from the leaf node that contains $p_i$ and we traverse $\mathcal{T}$ bottom up until we reach the root of the BBD tree. Let $v$ be a node we traverse; we update $w_i=w_i+v_s$.
After computing all values $w_i$,
for each color $c_j$ we find the $k_j$ points from $P(c_j)$ with the smallest weights $w_i$. Let $P_j$ be these points. For each $p_i\in P_j$ we set $\bar{x}_i=1$. Otherwise, if $p_i\notin P_j$ and $c(p_i)=c_j$, we set $\bar{x}_i=0$. If $\sum_{p_i\in P}w_i\bar{x}_i\leq 1$ the oracle returns $\bar{x}$ as a feasible solution. Otherwise, it returns that there is no feasible solution.

%In the technical report~\cite{techrep}, we show the correctness of our approach.
\paragraph{Proof of correctness.}
We show that the \textsf{ORACLE} we design is correct.
We first show that $\alpha_i=w_i$ so $h^\top Ax = \sum_{p_i\in P}w_ix_i$.
Recall that $h^\top Ax=\sum_{p_i\in P}\alpha_ix_i$, for the real coefficients $\alpha_i$. By definition, each $\alpha_i$ is defined as $\alpha_i=\sum_{p_\ell\in P}h[\ell]\cdot \mathcal{I}(p_i\in S_{p_\ell}^\eps)$, where $\mathcal{I}(p_i\in S_{p_\ell}^\eps)=1$ if $p_i\in S_{p_\ell}^\eps$ and $0$ otherwise.
If $p_i\in S_{p_\ell}^\eps$ then by definition $\alpha_i$ contains a term $h[\ell]$ in the sum. There exists also a node $u\in \mathcal{U}_{p_\ell}$
such that $p_i\in \square_u\cap P$, i.e.,  $p_i$ lies in a leaf node of the subtree rooted at $u$.
Starting from the leaf containing $p_i$, our algorithm will always visit the node $u$, and since $u\in \mathcal{U}_{p_\ell}$ we have that $h[\ell]$ is a term in the sum $u_s$ so by updating $w_i=w_i+u_s$ we include the term $h[\ell]$ in the weight $w_i$.
Overall, we have that $w_i=\alpha_i$ and our algorithm finds all the correct coefficients in the linear function $h^\top Ax$. Then we focus on minimizing the sum $\sum_{p_i\in P}w_ix_i$ satisfying $x\in \mathcal{P}$. For each color $c_j$ we should satisfy $\sum_{p_i\in P(c_j)}x_i\geq k_j$.
We can re-write $\sum_{p_i\in P}w_ix_i=\sum_{c_j\in C}\sum_{p_i\in P(c_j)}w_ix_i$.
Notice that the partial sum $\sum_{p_i\in P(c_j)}w_ix_i$ is minimized for $x\in\mathcal{P}$ setting $\bar{x}_i=1$ for the $k_j$ smallest factors $w_i$ in the partial sum.
Every point has a unique color, so there is no point $p_i$ that belongs in two different partial sums.
Repeating the same argument for each color $c_j\in C$, we conclude that indeed our algorithm finds the minimum value of $\sum_{p_i\in P}w_ix_i$ satisfying $x\in \mathcal{P}$. Overall, our algorithm correctly returns whether the feasibility problem $h^\top A x\leq 1$ for $x\in \mathcal{P}$ is feasible or infeasible.

\ansmeta{
Let $\bar{x}$ be the feasible solution returned by $\mathsf{Oracle}(\cdot)$.
Notice that by definition, $\bar{x}$ sets $k$ variables to $1$. Hence, for each Constraint~\eqref{eq:new}, it holds that $A_i\bar{x}-b_i\leq k-1$ and $A_i\bar{x}-b_i\geq -1$, where $A_i$ is the $i$-th row of $A$. Similarly, we can write $A_i\bar{x}\leq k$ and $A_i\bar{x}\geq 0$ since $b_i=1$. We conclude that our $\mathsf{Oracle}$ procedure computes a $k$-\textsf{ORACLE} as defined in~\cite{arora2012multiplicative}, so $\rho=k$.
}
\ansmeta{
\eparagraph{Example (cont)}
We show the execution of Algorithm~\ref{alg:alg1} in our example. Assume that $h^\top=[.1, .1, .1, .1, .4, .2]$, with $h_1+h_2+h_3+h_4+h_5+h_6=1$. By the definition of canonical subsets in the BBD tree, we have $\mathcal{U}_{p_1}=\{u^{(5)}, u^{(8)}\}$, $\mathcal{U}_{p_2}=\{u^{(5)}, u^{(8)}, u^{(3)}\}$, $\mathcal{U}_{p_3}=\{u^{(1)}, u^{(5)}\}$, $\mathcal{U}_{p_4}=\{u^{(6)}, u^{(10)}\}$, $\mathcal{U}_{p_5}=\{u^{(4)}\}$, and $\mathcal{U}_{p_6}=\{u^{(3)}, u^{(8)}, u^{(9)}\}$. Hence, from lines 3--5 we get $u^{(1)}_s=h_2+h_3=0.2$, $u^{(3)}_s=h_1=0.3$, $u^{(4)}_s=h_6=0.1$
, $u^{(5)}_s=h_4+h_5=0.4$, $u^{(6)}_s=h_4=0.05$, $u^{(7)}_s=h_6=0.1$, $u^{(8)}_s=h_1=0.3$, and all the rest $u^{(i)}_s=0$. Then in lines 6--10 we compute the coefficients $w_i$. For example consider $p_1$.
%The nodes in the path $p_1$ to the root are $u^{(10)}, u^{(5)}, u^{(2)}, u^{(0)}$.
We compute $w_1=u^{(10)}_s+u^{(5)}_s+u^{(2)}_s+u^{(0)}_s=0.4$. Similarly, we compute $w_2=0.5$, $w_3=0.9$, $w_4=0.1$, $w_5=0.5$, and $w_6=0.4$. We observe that these are indeed the correct coefficients in the inequality $h^\top Ax\leq 1$. For instance, the coefficient of $x_1$ is $h_1+h_2+h_3+h_4=0.4=w_1$. Overall, we have $h^\top A x= 0.4\cdot x_1+0.5\cdot x_2 + 0.9\cdot x_3 + 0.1\cdot x_4 + 0.5\cdot x_5 + 0.4\cdot x_6$ and we correctly identified all coefficients. Then in line 13 among the blue points we choose the smallest weight, $W_1=w_1=0.4$ and among the red points we choose the second smallest weight $W_2=w_6=0.4$. Hence, $P_1=\{p_1\}$ and $P_2=\{p_4, p_6\}$ and the algorithm sets $\bar{x}^\top=[1, 0, 0, 1, 0, 1]$.
Finally, in line 16 the algorithm computes $w_1+w_4+w_6=0.4+0.4+0.1=0.9<1$ so $\bar{x}$ is a feasible solution.
%It shows that for the particular probability vector $h$ as defined above $\mathsf{Oracle}(\cdot)$ chooses the points $\{p_1, p_5, p_6\}$ as a feasible solution.
}

\paragraph{Running time.}
For each new probability vector $h$ we construct $\mathcal{T}$ in $O(n\log n)$ time. For each point $p_i$ we find $\mathcal{U}_{p_i}$ in $O(\log n + \eps^{-d})$ time. Furthermore, the height of $\mathcal{T}$ is $O(\log n)$ so for each point $p_i$ we need additional $O(\log n)$ time to compute $w_i$.
After computing the weights, we find the smallest $k_j$ of them of each color in linear time. Overall, $Q_O=O(n\log n + n\eps^{-d})$.

%Overall, our algorithm finds a solution $\hat{x}$ such that $\hat{x}\in \mathcal{P}$ and $A_i\hat{x}\leq 1+\eps$ for each $i\leq n$. We apply $O(\frac{k\log n}{\eps^2})$ calls to the oracle and each oracle takes $O(n(\log n +\eps^{-d}))$ time. In the end of each iteration we spend $O(n(\log n + \eps^{-d}))$ time to update the probability vector $h$. Hence, the overall running time is $O\left(n\frac{k\log n}{\eps^2}(\log n + \eps^{-d})\right)$.

\mparagraph{The $\mathsf{Update}(\cdot)$ procedure}
\begin{algorithm}[t]
\ansmeta{
    \caption{\textsf{Update}$(P, \bar{x},\gamma, \eps)$}
    \label{alg:alg2}
    $\mathcal{T}\gets \text{ BBD tree on } P$\;
    \lForEach{$u\in \mathcal{T}$} {$u_w\gets 0$}
    \ForEach{$p_i\in P$ such that $\bar{x}_i>0$}{
         $v\gets \text{ leaf node of }\mathcal{T} \text{ such that } p_i\in \square_v\cap P$\;
        \ForEach{$u$ in the path from $v$ to the root of $\mathcal{T}$}{
            $u_w\gets u_w+\bar{x}_i$\;
        }
    }
    \ForEach{$p_\ell\in P$}{
        $\mathcal{U}_{p_\ell}\gets \mathcal{T}(p_\ell, \frac{\gamma}{2(1+\eps)})$\;
        $R_\ell=\sum_{u\in \mathcal{U}_{p_\ell}}u_w$\;
        $\delta_\ell=\frac{1}{k}(R_\ell-1)$\;
    }
    Update $h$ using $\delta_\ell$'s as described in~\cite{arora2012multiplicative}\;
    \Return $h$\;
}
\end{algorithm}
Next, we describe how we can update $h$ efficiently at the beginning of each iteration.
Let $\bar{x}$ be the solution of the oracle in the previous iteration.
Let $\delta_\ell=\frac{1}{k}(A_\ell\bar{x}-b_\ell)=\frac{1}{k}(A_\ell\bar{x}-1)$, where $A_\ell$ is the $\ell$-th row of $A$ ($\ell$-th constraint in~\eqref{eq:new}).
In~\cite{arora2012multiplicative} they update each $h[\ell]$ in constant time after computing $\delta_\ell$. In our case, if we try to calculate all $\delta_\ell$'s with a trivial way we would need $\Omega(nk)$ time leading to a $\Omega(nk^2)$ time overall algorithm. In Algorithm~\ref{alg:alg2} we show a faster way to calculate all $\delta_\ell$'s.
Our $\mathsf{Oracle}$ method sets $k$ variables $\bar{x}$ to $1$.
For each $p_\ell\in P$ the goal is to find $A_\ell\bar{x}=\sum_{p_i\in S_{p_\ell}^\eps}\bar{x}_i=\sum_{p_i\in S_{p_\ell}^\eps, \bar{x}_i> 0}\bar{x}_i$.
We modify $\mathcal{T}$ as follows. For each node $u\in \mathcal{T}$ we define the variable $u_w=0$. For each $p_i$ with $\bar{x}_i>0$ we start from the leaf containing $p_i$, and we visit the tree bottom up until we reach the root. For each node $u$ we encounter, we update $u_w=u_w+\bar{x}_i$. 
After the modification of $\mathcal{T}$, for each constraint/point $p_\ell$ we run the query $\mathcal{T}(p_\ell,\frac{\gamma}{2(1+\eps)})$ and we get the set of canonical nodes $\mathcal{U}_{p_\ell}$.
We compute $R_\ell\leftarrow\sum_{u\in \mathcal{U}_{p_\ell}}u_w=\sum_{p_i\in S_{p_\ell}^\eps}\bar{x}_i=A_\ell\bar{x}$. We return $\delta_\ell=\frac{1}{k}(R_\ell-1)$.
The correctness follows by observing that the coefficient of $p_i$'s variable in the $\ell$-th row of $A$ is $1$ if and only if $p_i\in S_{p_\ell}^{\eps}$.
We need $O(n)$ time to compute all values $v_w$ by traversing the tree $\mathcal{T}$ bottom up. Then for each $p_\ell$ we run a range query on $\mathcal{T}$ so we need $O(\log n + \eps^{-d})$. Overall, $Q_U=O(n\log n + n\eps^{-d})$.

\begin{algorithm}[t]
\ansmeta{
    \caption{\textsf{Round}$(P,\eps, \hat{x})$}
    \label{alg:algRound}
    $F\gets P$\;
    $\Xi\gets $ BBD tree for sampling on $F$\;
    \lForEach{$u\in\Xi$}{
        $u_b\gets 1$
    }
    \While{$F\neq \emptyset$}{
        $p_i\gets \Xi.\text{sample}()$, \hspace{0.25em} 
        $F\gets F\setminus\{p_i\}$\;
        $\mathcal{U}_{p_i}\gets \Xi(p_i,\frac{\gamma}{2(1+\eps)})$\;
        \If{$u_b==1$ for every $u\in \mathcal{U}_{p_i}$}{
            $S\gets S\cup \{p_i\}$\;
        }
        $v\gets \text{ leaf node of }\Xi \text{ such that } p_i\in \square_v\cap P$\;
        \lForEach{$u$ in the path from $v$ to the root of $\Xi$}{
            $u_b\gets 0$
        }
    }
    \Return $S$\;
}
\end{algorithm}

\ansmeta{
\eparagraph{Example (cont)}
We show the execution of the $\mathsf{Update}$ procedure using our example in Figure~\ref{fig:grid}, assuming that $\bar{x}^\top=[1, 0, 1, 0, 1, 0]$. By definition, $\delta_1=1$, $\delta_2=1$, $\delta_3=2$, $\delta_4=0$, $\delta_5=1$, and $\delta_6=0$. In lines 4--6, for $p_1$ we start from $u^{(10)}$ and we traverse the tree bottom-up. We set $u^{(10)}_w=u^{(5)}_w=u^{(2)}_w=u^{(0)}_w=\bar{x}[1]=1$. After traversing all points we have, $u^{(10)}_w=1$, $u^{(5)}_w=1$, $u^{(2)}_w=1$, $u^{(7)}_w=1$, $u^{(8)}_w=1$, $u^{(4)}_w=2$, $u^{(1)}_w=2$, $u^{(0)}_w=3$, and all the other nodes have weight $0$. Then in line 10, we compute $\delta_\ell$ for each $p_\ell$. For $p_1$, we get $\mathcal{U}_{p_1}=\{u^{(5)}, u^{(8)}\}$, so $R_1=u^{(5)}_w+u^{(8)}_w=2$, and $\delta_1=|2-1|=1$. Similarly, we compute the other $\delta_\ell$'s.
}

\mparagraph{The $\mathsf{Round}(\cdot)$ procedure}
\ansmeta{
The real vector $\hat{x}$ we get satisfies ($\mathsf{LP}$\ref{lp:2}) approximately. From the MWU method (see Theorem~\ref{thm:mutli-weights}) the Constraints in $\mathcal{P}$, (Constraints~\eqref{modeq1} and~\eqref{modeq3}) are satisfied exactly, however the Constraints~\eqref{eq:new} are satisfied approximately. In fact, it holds that
\begin{equation}
    \label{eq2Sol}
\sum_{p_i\in S_p^\eps} \hat{x}_i \leq 1+\eps,  \quad \quad \forall p\in P    
%\vspace{-0.5em}
\end{equation}
}
We follow a modified version of the randomized rounding from~\cite{addanki2022improved} to round $\hat{x}$ and return a set $S\subseteq P$ as the solution to the \fairDiv\ problem.
Our rounding method has major differences from the rounding in~\cite{addanki2022improved} (also briefly described in Section~\ref{sec:prelim}) because,  i) the Constraints~\eqref{eq2} are quite different from the Constraints~\eqref{eq2Sol} since the latter are satisfied with an additive error $\eps$, and their sum is over the set $S_{p}^\eps$, and ii) the rounding technique in~\cite{addanki2022improved} is executed in quadratic time with respect to the number of points. We propose a near-linear time algorithm to execute the rounding.

Before we describe the actual rounding algorithm we describe a modified BBD tree that we are going to use to sample efficiently.
For every point $p_i\in P$, we define its weight $\hat{x}_i$.
Let $\Xi$ be a BBD tree constructed over a weighted set $P$. We modify $\Xi$ so that we can sample a point $p_i$ with probability $\frac{\hat{x}_i}{\sum_{j\in F}\hat{x}_j}$, where $F\subseteq P$ is a subset of $P$.
For each node $u$ of $\Xi$, we store the value $u_s=\sum_{p_\ell\in \square_u\cap P} \hat{x}_\ell$, i.e., sum of $\hat{x}_i$'s of all points stored in the subtree of $u$.
Initially $F=P$.
We sample as follows: Assume any node $u$ having two children $v, e$. We visit $v$ with probability $v_s/(v_s+e_s)$ and we visit $e$ with probability $e_s/(u_s+e_s)$. When we insert a point $p_i$ in $F$, we start from the leaf node that stores $p_i$ and we visit the tree bottom up updating the values $u_s\leftarrow u_s-\hat{x}_i$ of the nodes $u$ we traverse. In this way, we ''remove`` the weight of point $p_i$ from the tree and it is not considered in the next iterations of sampling. It is straightforward to see that this procedure guarantees that each point $p_i$ is selected with probability $\frac{\hat{x}_i}{\sum_{j\in F}\hat{x}_j}$.
In order to make sure that we do not select two nearby points, for every node
$u$ of $\Xi$ we also set a boolean variable $u_b=1$.
If $u_b=1$ it means that our algorithm has not sampled a point that lies in $\square_u\cap P$. If $u_b=0$ it means that we have already sampled a point in $\square_u\cap P$ in a previous iteration so we should not re-consider node $u$ to get a new sample.

Let $S=\emptyset$ be the set of points that we return for the \fairDiv\ problem. Let also $F=P$ represent the set of points that we can sample from, as described in the previous paragraph.
Using $\Xi$ we sample a point $p_i$ with probability $\frac{\hat{x}_i}{\sum_{p_\ell\in F}\hat{x}_j}$. We update $F\leftarrow F\setminus\{p_i\}$.
Next, we run a query $\Xi(p_i, \frac{\gamma}{2(1+\eps)})$
%Next, we run a query $\Xi(p_i, \frac{\gamma}{2(1+\eps)^2})$ %\footnote{It will be clear later why we divide by $(1+\eps)^2$ instead of $1+\eps$.}
and we get a set of $O(\log n +\eps^{-d})$ canonical nodes $\mathcal{U}_{p_i}$.
%For each node $u\in \mathcal{U}_{p_i}$, we check whether $u_b=1$.
If $u_b=1$ for every $u\in \mathcal{U}_{p_i}$, then we insert $p_i$ in $S$. Otherwise, we do not insert it in $S$, i.e., we have already sampled a point from $S_{p_i}^\eps$ in a previous iteration.
Next, starting from the leaf node that contains $p_i$ we traverse the tree bottom up until we reach the root node. For each node $v$ we encounter we set $v_q=0$.
After we sample all points in $P$, we return the set $S$.
%With this implementation we create an ordering $\sigma$ of the points in $P$, where $\sigma(t)$ denotes the point sampled in the $t$-th iteration.
%Notice that we do not actually add in $S$ one point in every iteration.
%For the analysis below, for a point $p\in P$, we define $S_p^{1+\eps}=\{p_\ell\in P\mid \exists u\in \mathcal{U}_p^{1+\eps}, p_\ell\in \square_u\cap P\}$, i.e., the set of points covered by the nodes of the query $\Xi(p_i, \frac{\gamma}{2(1+\eps)^2})$. Recall that for every $p_\ell$ such that $||p-p+_\ell||\leq \frac{\gamma}{2(1+\eps)^2}$ then $p_\ell\in  S_p^{1+\eps}$ and if $p_\ell\in S_p^{1+\eps}$ then $||p-p_\ell||\leq \frac{\gamma}{2(1+\eps)}$.
%$P_i=\{p_j\in P\mid p_j\in \Xi(p_i, \frac{\gamma}{2(1+\eps)^2})\}$.

\begin{lemma}
\label{lem:divfair}
The minimum pairwise distance in $S$ is $\frac{\gamma}{2(1+\eps)}$ and 
for each color $c_j\in C$ it holds that $\expect[|S(c_j)|]\geq \frac{k_j}{1+\eps}$.
\end{lemma}
\begin{proof}
First, we show that the minimum pairwise distance in $S$ is $\frac{\gamma}{2(1+\eps)}$. Let $p_i, p_\ell$ be a pair of points with distance less than $\frac{\gamma}{2(1+\eps)}$. Without loss of generality assume that $p_i$ was added to $S$ first. Assume that $p_\ell$ is selected as a sample in a subsequent iteration. Since $||p_i-p_\ell||< \frac{\gamma}{2(1+\eps)}$, by definition, there exists a unique node $u\in \mathcal{U}_{p_\ell}$ such that $p_i\in \square_u\cap P$. Hence, $p_i$ lies in a leaf node of the subtree rooted at $u$. Since $p_i$ has already been selected we have that $u_b=0$ (because $u$ is an ancestor of the leaf node of $p_i$), so $p_\ell$ is not inserted in $S$.

Next, we argue about the fairness requirement.
Let $p_\ell$ be a point with $\hat{x}_\ell>0$. Let $V_t$ be the event that the first point included in $S$ from the set $S_{p_\ell}^{\eps}$ is the point from the $t$-th step.
Then,
\begin{align*}
\Pr[p_\ell\in S]&=\sum_{t=1}^{n}\Pr[\sigma(t)=p_\ell\mid V_t]\Pr[V_t] = \sum_{t=1}^{n}\frac{\hat{x}_\ell}{\sum_{p_i\in S_{p_\ell}^{\eps}}\hat{x}_i}\Pr[V_t]\\&=\frac{\hat{x}_\ell}{\sum_{p_i\in S_{p_\ell}^{\eps}}\hat{x}_i}\sum_{t=1}^{n}\Pr[V_t]
=\frac{\hat{x}_\ell}{\sum_{p_i\in S_{p_\ell}^{\eps}}\hat{x}_i}
\geq \frac{\hat{x}_\ell}{1+\eps}.
\end{align*}

The last inequality holds because $\sum_{p_i\in S_{p_\ell}^\eps}\hat{x}_i\leq 1+\eps$ from Constraints~\eqref{eq2Sol}. For $c_j\in C$, $\expect[|S(c_j)|]\geq \sum_{p_i\in P(c_j)} \frac{\hat{x}_i}{1+\eps}\geq \frac{k_j}{1+\eps}$.
\end{proof}

The rounding is executed in $Q_R=O(n(\log n +\eps^{-d}))$ time because $\Xi$ has $O(\log n)$ height and each query takes $O(\log n + \eps^{-d})$ time.

\ansmeta{
\eparagraph{Example (cont)}
We show the execution of the $\mathsf{Round}$ procedure using our example in Figure~\ref{fig:grid}.
Let $\hat{x}^\top=[.2, .2, .05, .15, .25, .15]$. Initially, for every node $u^{(i)}$ in Figure~\ref{fig:grid} (Middle), we have $u^{(i)}_b=1$. In the while loop (lines 4--10) we first sample a point from $P$. Let $p_2$ be the first point we sample. We get $\mathcal{U}_{p_2}=\{u^{(5)}, u^{(3)}, u^{(8)}\}$. All nodes in $\mathcal{U}_{p_2}$ have weight $1$ so in line 8, we add $p_2$ in $S$. Then, starting from $u^{(9)}$, we set all the weights of the nodes to $0$ until we reach the root. Hence, we set $u^{(9)}_b=u^{(5)}_b=u^{(2)}_b=u^{(0)}_b=1$. In the next iteration of the while loop, assume that we sample the point $p_6$. We have $\mathcal{U}_{p_6}=\{u^{(3)}, u^{(8)}, u^{(9)}\}$. We observe that $u^{(9)}_b=0$ so we do not add $p_6$ in $S$. Intuitively, $p_6$ is very close to $p_2$ that we have already added. Starting from $u^{(3)}$, we set $u^{(3)}_b=u^{(1)}_b=0$. In the next iteration, assume that we sample $p_5$. We get $\mathcal{U}_{p_5}=\{u^{(4)}\}$. In line 7, we observe that $u^{(4)}_b=1$ so we add $p_5$ in $S$. We also set $u^{(7)}_b=u^{(4)}_b=0$  (all the other nodes above $u^{(4)}$ have already weight equal to $0$). Next, assume that we sample $p_3$. We have $\mathcal{U}_{p_3}=\{u^{(5)}, u^{(1)}\}$. However, $u^{(5)}_b=u^{(1)}_b=0$ so we do not add $p_3$ in $S$. We also set $u^{(8)}_b=0$.
Next, assume that we sample $p_4$. We have $\mathcal{U}_{p_4}=\{u^{(6)}, u^{(10)}\}$ with $u^{(6)}_b=u^{(10)}_b=1$. So we add $p_4$ in $S$ and we set $u^{(6)}_b=0$.
In the final iteration, we sample $p_1$. We get $\mathcal{U}_{p_1}=\{u^{(5)}, u^{(8)}\}$ and we observe that $u^{(5)}_b=0$ so we do not add $p_1$ in $S$. We set $u^{(10)}_b=0$.
Our algorithm returns $S=\{p_2, p_5, p_4\}$.
}

\ansmeta{
\mparagraph{Compute the set $\Gamma$}
We note that so far we assumed that $\gamma$ is any distance and we tried to find a solution $S$ with $\diversity(S)\geq \frac{\gamma}{2(1+\eps)}$. In order to find a good approximation of the optimum  diversity $\gamma^*$, we use the notion of the Well Separated Pair Decomposition (WSPD)~\cite{callahan1995decomposition, har2005fast} briefly described in Section~\ref{sec:prelim}. Let $\Gamma$ be the sorted array of $O(n/\eps^d)$ distances from WSPD.
Any pairwise distance in $P$ can be approximated by a distance in the array $\Gamma$ within a factor $1+\eps$, hence, we might not get the optimum $\gamma^*$ exactly. In the worst case, we might get a smaller value which is at least $\frac{\gamma^*}{1+\eps}$.
We need $O(n\eps^{-d}\log n)$ time to compute and sort the WSPD distances, so $Q_{\Gamma}=O(n\eps^{-d}\log n)$.
}
%The binary search executes $O(\log n)$ iterations.\footnote{For simplicity we consider $\eps\geq 1/n$.} Hence, the running time of the overall algorithm increases by a $\log n$ factor.

Putting everything together,
%and setting $\eps\leftarrow \eps/4$,\footnote{Notice that $\frac{\gamma}{2(1+\eps/4)^2}\geq \frac{\gamma}{2(1+\eps)}$ for any $\eps\in(0,1)$.}
we get the next theorem.

\begin{theorem}
\label{thm:main}
Let $C$ be a set of $m$ colors, $P$ be a set of $n$ points in $\Re^d$ for a constant $d$, where each point $p\in P$ is associated with a color $c(p)\in C$ and let $k_1, \ldots, k_m$ be $m$ integer parameters with $k_1+\ldots +k_m=k$.
Let $\eps\in (0,1)$ be a constant.
There exists an algorithm for the \fairDiv\ problem that returns a set $S\subseteq P$ such that $\expect[|S(c_j)|]\geq \frac{k_j}{1+\eps}$ for each color $c_j\in C$, and $\diversity(S)\geq \frac{\gamma^*}{2(1+\eps)}$ in $O(nk\log^3 n)$ time and $O(n)$ space.
\end{theorem}

\begin{comment}
\begin{theorem}
\label{thm:main}
Let $C$ be a set of $m$ colors, $P$ be a set of $n$ points in $\Re^d$ for a constant dimension $d$, where each point $p\in P$ is associated with a color $c(p)\in C$ and let $k_1, \ldots, k_m$ are $m=O(k)$ integer parameters with $k_1+\ldots +k_m=k$. Let $\eps$ be a parameter. There exists an algorithm that returns a set $S\subseteq P$ such that $\expect[|S(c_j)|]\geq \frac{k_j}{1+\eps}$ for each color $c_j$, and $\min_{p,q\in S}||p-q||\geq \frac{\gamma}{2(1+\eps)}$ in $O\left(n\frac{k\log^2 n}{\eps^2}(\log n +\eps^{-d})\right)$ time.
\end{theorem}
\end{comment}

\subsection{Diversity with high probability}
\label{sub:prob}
In this subsection, we describe how to use the result of the previous subsection to get a solution where the fairness constraints hold with probability at least $1-\delta$ for any parameter $\delta\in(0,1)$.
We get the intuition from~\cite{addanki2022improved}, however the implementation of our algorithm has significant changes to achieve the near-linear time execution.

We transform the near-optimal fractional solution $\hat{x}$ we derived from the MWU method to a new fractional solution $\hat{y}\in \Re^n$.
Recall that $\mathcal{T}$ is a BBD tree constructed on $P$.
Recall that $\mathcal{U}(p,\beta)$ is the set of canonical nodes returned by query $\mathcal{T}(p, \beta)$.
Let $S_p(\beta)\subseteq P$ be the set of points covered by the rectangles associated with the canonical nodes $\mathcal{U}(p,\beta)$.
Notice that $S_p^{\eps}=S_p(\frac{\gamma}{2(1+\eps)})$.

%$P_{p}(\eps,\gamma/6)$ be the set of points in $\mathcal{T}(p,\frac{\gamma}{6(1+\eps)^2})$.

The feasible solution $\hat{y}$ we will construct is the solution for the following set of (non-linear) constraints.

%\noindent\begin{minipage}[t][0ex][t]{0.7cm}
%\lpproblem
%\end{minipage}\vspace* {-2ex}
\begin{align}
\label{eq4.1}\sum_{p_i\in P(c_j)} y_i&\geq k_j,  \quad\quad \forall j\in [m]\\
\label{eq4.2} \sum_{p_i\in S_p\left(\frac{\gamma}{6(1+\eps)^2}\right)} y_i& \leq 1+\eps,  \quad \quad \forall p\in P\\
\label{eq4.3}y_{i} &\geq 0, \quad \quad \forall p_i\in P\\
\label{eq4.4}y_i>0\text{ and }y_\ell&>0 \Rightarrow ||p_i-p_{\ell}||\geq\frac{\gamma}{3(1+\eps)^2}, \\
\nonumber&\quad\quad\quad\quad\forall p_i,p_\ell\in P(c_j), \forall j\in[m]
\end{align}

%\begin{equation*}
%\begin{array}{ll@{}ll}
%\displaystyle\sum\limits_{p_i: c(p_i)=j}& y_i\geq k_j,  \quad\forall j\in [m]\\
%\displaystyle\sum\limits_{p_i\in P_{p}(\eps,\gamma/6)}& y_i \leq 1+\eps,  \quad \forall p\in P\\
%&y_{i} \geq 0, \quad \forall p_i\in P
%\\
%&y_i>0\text{ and }y_\ell>0 \Rightarrow ||p_i-p_{\ell}||_2\geq\frac{\gamma}{3(1+\eps)^2}, \\&\quad\quad\quad\forall p_i,p_\ell\in P(c_j), \forall j\in[m]
%\end{array}
%\end{equation*}

Next, we describe the algorithm to get a solution $\hat{y}$ for the new set of constraints.
For each color $c_j\in C$ we construct a BBD tree $\mathcal{T}_j$ on $P(c_j)$. For a query $\mathcal{T}_j(p,\beta)$, we define $\mathcal{U}_j(p,\beta)$ and $S_{p,j}(\beta)$ as we had before, considering only the points of color $c_j$. Each point $p_i$ in $\mathcal{T}_j$ is associated with its weight $\hat{x}_i$. In each node $u$ of every tree $\mathcal{T}_j$ we store the sum $u.s$ of the weights of the points stored in the subrtee rooted at $u$, i.e., $u_s=\sum_{p_i\in \square_u\cap P(c_j)}\hat{x}_i$.
Given a query ball $\mathcal{B}(p,r)$ we can run a weighted range sum query in $\mathcal{T}_j$ and return $M_j(p,r)=\sum_{p_i\in S_{p,j}(r)}\hat{x}_i$
in $O(\log n +\eps^{-d})$ time.
If $u.s>0$, then we call node $u$ \emph{active}. Otherwise, if $u_s=0$, node $u$ is \emph{inactive}.
%As we describe in the algorithm, we only care about the total sum of the points in the active areas.

For each point $p_i$ with $\hat{x}_i>0$ and no inactive ancestor, we perform the following steps. Let $c_j$ be the color of $p_i$.
We run a sum query on $\mathcal{T}_j$ and we get  $\hat{y}_i=M_j(p_i,\frac{\gamma}{3(1+\eps)^2})$. 
%Let $\mathcal{T}_j(p_i,\frac{\gamma}{3(1+\eps)})$ be the set of canonical nodes in the query.
During the search procedure if we visit a node $u$ with $u.s=0$, we stop traversing the subtree rooted at $u$. For each node $u\in \mathcal{U}_j(p_i,\frac{\gamma}{3(1+\eps)^2})$, we traverse the tree from $u$ to the root of the tree and for each node $v$ we visit, we update $v.s\leftarrow v.s-u.s$. Then we set $u.s=0$.
In the end of this process, if there are $\hat{y}_i$ values that are not defined, we set $\hat{y}_i=0$.

\mparagraph{Correctness}
We show that our method finds a solution $\hat{y}$ that satisfies the new constraints.
Constraints~\eqref{eq4.1} $\sum_{p_i\in P(c_j)} \hat{y}_i\geq k_j$ are satisfied because for each color $c_j\in C$, $\sum_{p_i\in P(c_j)} \hat{y}_i = \sum_{p_i\in P(c_j)} \hat{x}_i$.
This equality holds because when we find  node $u$ as canonical node for a query $\mathcal{T}_j(p, \frac{\gamma}{3(1+\eps)^2})$, we set $u$ to an inactive node so no point $q\in \square_u\cap P(c_j)$ contributes to another another variable $\hat{y}_\ell$ in the next iterations.
For Constraints~\eqref{eq4.2}, for every point $p\in P$, we have,
\begin{align*}
\sum\limits_{p_i\in S_p(\frac{\gamma}{6(1+\eps)^2})} \hat{y}_i&\leq
\sum_{p_\ell\in P: \exists p_i\in S_p(\frac{\gamma}{6(1+\eps)^2}), p_\ell\in S_{p_i}(\frac{\gamma}{3(1+\eps)^2})}\hat{x}_\ell
\\&
\leq \sum_{p_\ell\in P: ||p-p_\ell||\leq \frac{\gamma}{2(1+\eps)}}\hat{x}_\ell 
\leq \sum\limits_{p_\ell\in S_{p}(\frac{\gamma}{2(1+\eps)})}\hat{x}_\ell\\&=\sum_{p_\ell\in S_p^\eps} \hat{x}_\ell\leq 1+\eps.
\end{align*}
Constraints~\eqref{eq4.3} are satisfied by definition. Finally, for Constraints~\eqref{eq4.4}, notice that when our algorithm sets some positive value $\hat{y}_i$ then all canonical nodes in $\mathcal{U}_j(p_i,\frac{\gamma}{3(1+\eps)^2})$ are set to inactive. Hence, any point $p_\ell$ with $||p_i-p_\ell||< \frac{\gamma}{3(1+\eps)^2}$, has at least one inactive ancestor and $\hat{y}_i$ is set to $0$.

\mparagraph{Running time}
We construct all BBD trees $\mathcal{T}_j$ in $O(n\log n)$ time.
For each point $p_i$ we check if there is no inactive ancestor in $O(\log n)$ time since $\mathcal{T}_j$ has depth $O(\log n)$. Each weighted range sum query on $\mathcal{T}_j$ takes $O(\log n + \eps^{-d})$ time.
The $u.s$ variables are also updated in $O(\log n + \eps^{-d})$ time.
Overall, given the values $\hat{x}$, we can compute $\hat{y}$ in $O(n(\log n + \eps^{-d}))$ time.

%Due to space limit, we show the rounding in the technical report~\cite{techrep}. Putting everything together we get the next theorem.

\paragraph{Rounding}
The rounding procedure is similar to the rounding in Subsection~\ref{sub:exp}. Again, we construct a BBD tree $\Xi$ to sample points from $P$. Each time we choose a point $p_i$ we check if another point has already been selected in a previous iteration in the area of $\Xi(p_i,\frac{\gamma}{6(1+\eps)^3})$\footnote{For technical reasons we set the denominator to $(1+\eps)^3$ instead of $(1+\eps)^2$.}.
Using the proof of Lemma~\ref{lem:divfair}, we get that for any pair $p_i, p_j\in S$ in the set we return, we have $\min_{p,q\in S}||p-q||_2\geq \frac{\gamma}{6(1+\eps)^3}$. Furthermore, let $Y_i$ be the random variable which is $1$ if $p_i\in S$ and $0$ otherwise. Using the same proof, we get $\expect\left[\sum_{p_i\in S(c_j)}Y_i\right]\geq \frac{k_j}{1+\eps}$ for each color $c_j$. By definition, all variables $Y_i$ are independent, so we can apply the Chernoff bound to get a bound on the number of points of each color with high probability. In particular, if $k_j\geq 3\frac{1+\eps}{\eps^2}\log(2m)$, we get $\Pr[\sum_{p_i\in S(c_j)}Y_i\leq (1-\eps)\frac{k_j}{1+\eps}]\leq \frac{1}{2m}$. If we apply the union bound we have that $\sum_{p_i\in S(c_j)}Y_i\geq (1-\eps)\frac{k_j}{1+\eps}$ for all colors $c_j\in C$ with probability at least $1/2$. If we repeat the process $\log \frac{1}{\delta}$ iterations, then with probability at least $1-\delta$ we find a solution $S$ with $\sum_{p_i\in S(c_j)}Y_i\geq (1-\eps)\frac{k_j}{1+\eps}$ for all colors $c_j\in C$. In each repetition we spend $O(n(\log n + \eps^{-d}))$ as in Subsection~\ref{sub:exp}. The total running time of the rounding procedure is $O(n\log(\frac{1}{\delta})(\log n + \eps^{-d}))$.

Finally, as we had in Subsection~\ref{sub:exp}, recall that we run a binary search on the WSPD set $L$, finding a value $\gamma$ such that $\gamma^*\geq \gamma\geq \gamma^*/(1+\eps)$.

Putting everything together, and setting $\eps\leftarrow \eps/6$,\footnote{Notice that $\frac{\gamma}{6(1+\eps/6)^4}\geq \frac{\gamma}{6(1+\eps)}$ and $\frac{1-\eps/6}{1+\eps/6}k_j\geq \frac{k_j}{1+\eps}$ for $\eps\in(0,1)$.} we get the next theorem.

\begin{theorem}
\label{thm:main2}
Let $C$ be a set of $m$ colors, $P$ be a set of $n$ points in $\Re^d$ for a constant dimension $d$, where each point $p\in P$ is associated with a color $c(p)\in C$ and let $k_1, \ldots, k_m$ be $m$ integer parameters with $k_1+\ldots +k_m=k$ and $k_j\geq 3(1+\eps)\eps^{-2}\log(2m)$ for each $c_j\in  C$. Let $\eps$ be a constant parameter and $\delta$ be a parameter. There exists an algorithm that returns a set $S\subseteq P$ such that $|S(c_j)|\geq \frac{k_j}{1+\eps}$ for each color $c_j\in C$ with probability at least $1-\delta$, and $\diversity(S)\geq \frac{\gamma}{6(1+\eps)}$ in 
$O\left(nk\log^3 n + n\log(\frac{1}{\delta})\log n\right)$ time and $O(n)$ space.
%$O\left(n(\frac{k\log^2 n}{\eps^2}+\log\frac{1}{\delta})(\log n +\eps^{-d})\right)$ time.
\end{theorem}

\section{Coreset}
\label{sec:coreset}
As shown in Section~\ref{sec:prelim},
in~\cite{addanki2022improved} they describe a $(1+\eps)$-coreset for the \fairDiv\ problem. In particular, using the Gonzalez algorithm~\cite{gonzalez1985clustering} for the $k$-center clustering problem, they get a set $G\subseteq P$ of $O(mk\eps^{-d})$ points in $O(nk\eps^{-d})$ time, such that $G$ contains a solution for the \fairDiv\ problem with diversity at least $\gamma^*/(1+\eps)$.
Their proof of correctness relies on the execution of Gonzalez algorithm choosing the furthest point from the set of centers that have been already selected in each iteration. Unfortunately, Gonzalez algorithm takes $O(nk)$ time. Ideally, we would like to use other faster constant approximation algorithms for the $k$-center problem in the Euclidean case.
%Ideally, these algorithms can also be used in the streaming or the query range setting
%, for example the algorithm by Feder and Green~\cite{} that runs in $O(n\log k)$ time or the algorithm by Har-Peled and Raichel~\cite{} that runs in $O(n)$ time in expectation.
In this section, we show a more general coreset construction. We show that if $k'=O(\eps^{-2d}k)$ points are chosen for each color $c_j\in C$ using any constant approximation algorithm for the $k'$-center clustering problem in the Euclidean space, then their union is a valid coreset for the \fairDiv\ problem.

Let \Alg\ be an $\alpha$-approximation for the $k'$-center clustering problem that runs in $O(T(n,k'))$ time. For every color $c_j\in C$, we run \Alg\ on $P(c_j)$ for $k'=O(\eps^{-2d}k)$ and we get the set of centers $G_j'$. We return the coreset $G'=\bigcup_{c_j\in  C}G_j'$. The coreset $G'$ is constructed in $O(\sum_{c_j\in C}T(|P(c_j)|,k'))$ time and has cardinality $|G'|=O(\eps^{-2d}km)$.

\begin{lemma}
\label{lem:newCoreset}
The set $G'$ is a $(1+\eps)$-coreset for the \fairDiv\ problem.
\end{lemma}
\begin{proof}
We fix a color $c_j\in C$. For any $k$, let $\mu_k$ be the optimum radius for the $k$-center clustering problem in $P(c_j)$.
Let $\hat{k}=O(\eps^{-d}k)$ and $k'=O(\eps^{-2d}k)$. For a subset $Q\subseteq P$, we define $\mu(Q)=\max_{p\in P(c_j)}\min_{q\in Q}||p-q||$, i.e., the value of the $|Q|$-center solution $Q$ on $P(c_j)$.
%Let $r_\alpha$ be the value of the $k$-objective (maximum distance from a point to its closest center) returned by \Alg\ on $P(c_j)$.
Let $\xi$ be a constant number. Let $D$ be a grid in $\Re^d$ such that each grid cell has side length $\frac{\eps\cdot \mu_{\hat{k}}}{4\alpha\cdot \xi}$. Let $A=\emptyset$. For every cell $g\in D$, if $|g\cap P(c_j)|>0$, then we insert a representative point $p_g\in D\cap P(c_j)$ in $A$. It is known~\cite{agarwal2002exact, abrahamsen2017range} that $|A|=O(\eps^{-2d}k)$ and for every point $p\in P(c_j)$ there exists a point $p_a\in A$ such that $\frac{\eps}{4\alpha}\mu_{\hat{k}}$. Hence, $\mu(A)\leq \frac{\eps}{4\alpha}\mu_{\hat{k}}$. By definition it also holds that $\mu_{k'}\leq \mu(A)$.
It follows that
$\mu_{k'}\leq \frac{\eps}{4\alpha}\mu_{\hat{k}}$.
We have,
    $\mu(G'_j)\leq \alpha\mu_{k'}\leq \frac{\eps}{4}\mu_{\hat{k}}$.
For any $k$, let $\sigma_k$ be the minimum pairwise distance of the optimum solution of the (unfair) $k$-MaxMin diversification problem in $P(c_j)$ (i.e., choose a set of $k$ points in $P(c_j)$ that maximize the minimum pairwise distance). It is always true that $\sigma_k\geq \mu_k$~\cite{tamir1991obnoxious}. It is also known that the Gonzalez algorithm for $k$ iterations returns a solution with diversity at least $\frac{1}{2}\sigma_k$~\cite{tamir1991obnoxious}.

Let $O=\{o_1,\ldots, o_{\hat{k}}\}$ be the list of $\hat{k}$ centers returned by the Gonzalez algorithm (in order) on $P(c_j)$. For any $o_i\in O$, let $p_i\in G_j'$ be its closest point in $G_j'$ and let $P_j'=\bigcup_{i\in[\hat{k}]} p_i$, and $P'=\bigcup_{c_j\in C}P_j'$.
We show that $P'\subseteq G'$ is a valid $(1+\eps)$-coreset.
Let $r=||o_i-o_{\ell}||$ for any pair $o_i, o_\ell\in O$ %(actually it is sufficient to consider only the pair with the smallest distance which is $o_{\hat{k}-1}, o_{\hat{k}}$).
We have $||p_i-p_{\ell}||\geq r-||o_i-p_i||-||o_{\ell}-p_{\ell}||\geq r -\frac{\eps}{2}\mu_{\hat{k}}\geq r-\frac{\eps}{2}\sigma_{\hat{k}}\geq (1-\eps)r$.
The last inequality holds because $r\geq \frac{1}{2}\sigma_k$ (recall that Gonzalez algorithm returns a set of points with diversity at least $\frac{1}{2}\sigma_k$).
Similarly, $||p_i-p_{i+1}||\leq (1+\eps)r$. Hence, $(1-\eps)r\leq ||p_i-p_{i+1}||\leq (1+\eps)r$. All inequalities from Theorem 5 in~\cite{addanki2022improved} are satisfied within a $(1-\eps)$ (or $(1+\eps)$) factor, so by setting $\eps\leftarrow\eps/\zeta$, for a sufficiently large constant $\zeta$ that depends on $d$, we conclude that $G'$ is an $(1+\eps)$-coreset for \fairDiv.
\end{proof}

%In the Euclidean space we can use~\cite{har2015net} to run the same algorithm in $O(n)$ expected time.
%Using our result from Theorem~\ref{thm:main} and choosing the correct $\eps$ (in particular $\eps\leftarrow \eps/6$), we can get the following corollary.

%Putting everything together, we proved the following result.
Overall, we state our new result.
\begin{theorem}
    \label{thm:coreset}
   In the Euclidean space, any constant approximation algorithm for the $k$-center clustering with running time $O(T(n,k))$ can be used to derive a $(1+\eps)$-coreset for the \fairDiv\ problem of size $O(\eps^{-2d}mk)$ in $O(\sum_{c_j\in C}T(|P(c_j)|,\eps^{-2d}k))$ time.
\end{theorem}

In the next result we fix \Alg\ to be either the Feder and Greene~\cite{feder1988optimal} algorithm or the Har-Pelled and Raichel algorithm~\cite{har2015net} to return a $2$-approximation for the $k$-center clustering in $O(n\log k)$ time or $O(n)$ expected time, respectively.
Using our coreset $G'$ as input to the algorithm in Theorem~\ref{thm:main} we get the next corollary.

\begin{corollary}
\label{cor:main}
%Let $C$ be a set of $m$ colors, $P$ be a set of $n$ points in $\Re^d$ for a constant dimension $d$, where each point $p\in P$ is associated with a color $c(p)\in C$ and let $k_1, \ldots, k_m$ be $m$ integer parameters with $k_1+\ldots +k_m=k$. Let $\eps\in(0,1)$ be a constant parameter.
There exists an algorithm that returns a set $S\subseteq P$ such that $\expect[|S(c_j)|]\geq \frac{k_j}{1+\eps}$ for each color $c_j\!\in\!C$, and $\diversity(S)\geq \frac{\gamma^*}{2(1+\eps)}$ in $O(n\log k+mk^2\log^3(k))$ time and $O(n)$ space.
The same algorithm can be executed in $O(n+mk^2\log^3(k))$ expected time.
\end{corollary}

Using our coreset $G'$ as input to the algorithm in Theorem~\ref{thm:main2} we get the next corollary.

\begin{corollary}
\label{cor:main2}
%Let $C$ be a set of $m$ colors, $P$ be a set of $n$ points in $\Re^d$ for a constant dimension $d$, where each point $p\in P$ is associated with a color $c(p)\in C$ and let $k_1, \ldots, k_m$ be $m$ integer parameters with $k_1+\ldots +k_m=k$ and $k_j\geq 3(1+\eps)\eps^{-2}\log(2m)$ for each $j\in [m]$. Let $\eps\in(0,1)$ be a constant parameter, and $\delta\in(0,1)$ be another parameter.
There exists an algorithm that returns a set $S\subseteq P$ such that $|S(c_j)|\geq \frac{k_j}{1+\eps}$ with probability at least $1-\delta$ for each color $c_j\in C$, and $\diversity(S)\!\geq\! \frac{\gamma^*}{6(1+\eps)}$ in $O(n\log k+mk^2\log^3 k+mk\log k\log\frac{1}{\delta})$ time and $O(mk)$ additional space.
The same algorithm can be executed in $O(n+mk^2\log^3 k+mk\log k\log\frac{1}{\delta})$ expected time.
\end{corollary}

We notice that the running time of the algorithm in Corollary~\ref{cor:main} or Corollary~\ref{cor:main2} is not always faster than the algorithm in Theorem~\ref{thm:main} or Theorem~\ref{thm:main2}, respectively. While for small values of $k$, an algorithm using the coreset is faster, when $k$ is large the asymptotic running time of the algorithm without using the coreset is faster.

\section{Extensions}
\label{sec:extensions}
In this section, we show how our coreset construction and our new algorithms for the \fairDiv\ problem can be used to get a faster algorithm in the streaming setting, \streamFairDiv\ problem. We also show how our results can be used to design the first efficient data structure for the \queryFairDiv\ problem.

\subsection{Streaming setting}
In the streaming setting we care about three quantities: The number of elements that the algorithm needs to store in each iteration, the update time for any new item we get, and the post-processing time we need to create an actual solution for the \fairDiv\ problem. %considering all the items we have seen so far.

There are two known algorithms for the \streamFairDiv\ problem.
In~\cite{wang2022streaming}, they describe a streaming algorithm that stores $O(mk\log\Delta)$ elements in memory, takes $O(k\log\Delta)$ time per element for streaming processing, and $O(m^2k^2\log \Delta)$ time for post-processing that returns a $\frac{1-\eps}{3m+2}$-approximation solution. The streaming processing time per element can be improved to $O(\log k\cdot\log\Delta)$ in the Euclidean space using an efficient dynamic data structure for the closest pair problem~\cite{bespamyatnikh1995optimal}, as described in~\cite{agarwal2020efficient}.
In~\cite{addanki2022improved}, they improved the approximation factor to a constant, however all asymptotic complexities still depend on $\log\Delta$.
Notice that $\Delta$ can be exponential with respect to $n$ even in $\Re^d$ making all quantities linear on $n$ in the worst case. We present the first constant-approximation streaming algorithm, called \newStreamAlg{}, for the \fairDiv\ problem whose space, update, and post-processing time are independent of $\Delta$. %In fact, our new streaming algorithm dominates the algorithm from~\cite{wang2022streaming} in all quantities and it returns a constant approximation (independent of $m$) satisfying the fairness constrains approximately by a factor of $1+\eps$.

Because of Lemma~\ref{lem:newCoreset} it is sufficient to maintain any constant approximation of the $k$-center clustering solution over the set of items we have seen so far.
It is known that the doubling algorithm~\cite{charikar1997incremental} maintains an $8$-approximation for the $k$-center problem in $\Re^d$ with $O(k\log k)$ update time (using an efficient dynamic data structure for the closest pair problem~\cite{bespamyatnikh1995optimal}, as described in~\cite{agarwal2020efficient}) storing $O(k)$ elements~\cite{agarwal2020efficient}.
%There are some recent results that can improve the update time to $O(\log k)$~\cite{agarwal2020efficient} using geometric tools and known results from~\cite{guha2009tight, matthew2008streaming}.
Using the doubling algorithm to maintain a constant approximation for the $k$-center clustering, our new coreset construction in Theorem~\ref{thm:coreset}, and our new near-linear time algorithm in Theorem~\ref{thm:main} for the \fairDiv\ problem we give the following result for the \streamFairDiv\ problem.
%Recall that we assume $m=O(k)$.
\begin{theorem}
\label{thm:streaming}
For a constant parameter $\eps$, there exists a streaming algorithm for the \streamFairDiv\ problem that stores $O(mk)$ items, has $O(k\log k)$ update time, and has $O(mk^2\log^3 k)$ post-processing time. After the post-processing the algorithm returns a set of points $S$ such that $\expect[S(c_j)]\geq \frac{k_j}{1+\eps}$ for each color $c_j\in C$, and $\diversity(S)\geq \frac{\gamma^*}{2(1+\eps)}$. 
\end{theorem}
%Using the result from Theorem~\ref{thm:main2} (instead of Theorem~\ref{thm:main}) we can also satisfy the fairness constraints with high probability.
In the streaming setting, our new algorithm is called \newStreamAlg{}.

\subsection{Range-query setting}
%In the range-query setting (\queryFairDiv\ problem), given a set of points $P$, the goal is to construct a data structure $\mathcal{D}$, such that given a query rectangle $R$ and parameters $k_1,\ldots, k_j$, the goal is to return a solution for the \fairDiv\ problem in $P\cap R$ efficiently (in sublinear time with respect to $n$). There is no known data structure for the \queryFairDiv\ problem.

%Using our new results we design the first efficient data structure for the \queryFairDiv\ problem.
Given a set of $n$ points $P\in \Re^d$, in~\cite{abrahamsen2017range, oh2018approximate} they show
that there exists a data structure of $O(n\log^{d-1} n)$ space that can be constructed in $O(n\log^{d-1} n)$ time, such that, given any query hyper-rectangle $R$ and any query parameter $k$, a $(2+\eps)$-approximation of the $k$-center clustering in $P\cap R$ is returned in $O(k\log^{d-1}n + \eps^{-d}k\log\frac{k}{\eps})$ time. Using this $k$-center data structure, we construct our coreset from  Theorem~\ref{thm:coreset} in $P\cap R$ efficiently, and then using our new near-linear time algorithm in Theorem~\ref{thm:main} for the \fairDiv\ problem, we give the following result for the \queryFairDiv\ problem.

\begin{theorem}
\label{thm:query}
For the \queryFairDiv\ problem, a data structure of size $O(n\log^{d-1} n)$ can be constructed in $O(n\log^{d-1} n)$ time, such that given a query rectangle $R$, a constant parameter $\eps$, and parameters $k_1,\ldots, k_m$ with $k_1+\ldots k_m=k$, it returns a set $S\subseteq P\cap R$ in $O(mk\log^{d-1}n + mk^2\log^3 k)$ time such that $\expect[S(c_j)]\geq \frac{k_j}{1+\eps}$ for each color $c_j\in C$, and $\diversity(S)\geq \frac{\gamma^*}{2(1+\eps)}$.
\end{theorem}

Using the result from Theorem~\ref{thm:main2} (instead of Theorem~\ref{thm:main}) we can also satisfy the fairness constraints  with high probability.

\section{Experiments}
\label{sec:experiments}
In this section, we evaluate the effectiveness of our algorithm to identify a diverse and fair set of points. Specifically, we answer the following research questions:\\
\noindent \textbf{RQ1:} How does \newAlg{} behavior change with varying parameters? Is the fairness requirement violated by \newAlg{}? This RQ identifies the best parameters for \newAlg{} to be used for comparison with baselines.\\
\noindent \textbf{RQ2:} How does \newAlg{} result compare against other baselines? Is it efficient to identify a fair and diverse solution?%\\
%\noindent \textbf{RQ3:} Is MWU scalable to million scale datasets? How does the running time increase with increasing dataset size $n$ and $k$?

\noindent \textbf{RQ3:} How efficient is our new algorithm \newStreamAlg{} in the streaming setting?

%\begin{itemize}
%    \item Among algorithms that return similar diversity with our new MWU-based algorithm, our algorithm is always the fastest one.
%    \item If another algorithm is faster than our new algorithm, then its diversity is very low.
%    \item Even if in theory our new algorithm might not return all $k_j$ points from each color $c_j\in C$, in almost all cases we satisfy all fairness constraints.
%    \item Overall, the fastest algorithm that returns fair sets with good enough diversity is our new MWU-based algorithm.
%\end{itemize}

\begin{table}[ht]
    \centering
   % \resizebox{0.9\linewidth}{!}{
        \begin{tabular}{ |c|c|c|c|c| } 
            \hline
            Dataset & Groups & m & d & n\\
            \hline
            Adult & Race, Sex & 10 & 6 & 32,561\\\hline
            Diabetes & Sex, Meds Prescribed-Y/N & 4 & 8 & 101,763\\\hline
            Census & Sex, Age & 14 & 6 & 2,426,116\\\hline
            Popsim & Race & 5 & 2 & 4,110,608\\\hline
            Popsim\_1M & Race & 5 & 2 & 821,804\\\hline
            Beer reviews & Style of beer & 3 & 6 &1,518,829\\\hline
           % \ansmeta{Yelp} & \ansmeta{Cuisines} & & &\\\hline
        \end{tabular}
       % }
    \caption{Dataset Statistics}\label{table:dataset_statistics}
\end{table}

\begin{comment}
\begin{table}[ht]
    \caption{Dataset Statistics}
    \centering
        \begin{tabular}{ |c|c|c|c|c| } 
            \hline
            Dataset & Groups & m & d & n\\
            \hline
            
            \multirow{1}{4em}{Adult}    & \makecell{Race\\ Sex} & 10 & 6 & 32,561\\ 
                                        \hline
            \multirow{1}{4em}{Census}   & \makecell{Sex\\ Age} & 14 & 6 & 2,426,116\\
                                        \hline
            \multirow{1}{4em}{Diabetes} & \makecell{Sex\\ Meds Prescribed-Yes/No} & 4 & 8 & 101,763\\
                                        \hline
            \multirow{1}{4em}{Popsim}   & Race & 5 & 2 & 4,110,608\\ 
                                        \hline
            \multirow{1}{4em}{Popsim}   & Race & 5 & 2 & 821,804\\ 
                                        \hline
        \end{tabular}
    \label{table:dataset_statistics}
\end{table}
\end{comment}

    \begin{figure*}[t]
    \includegraphics[width=\textwidth]{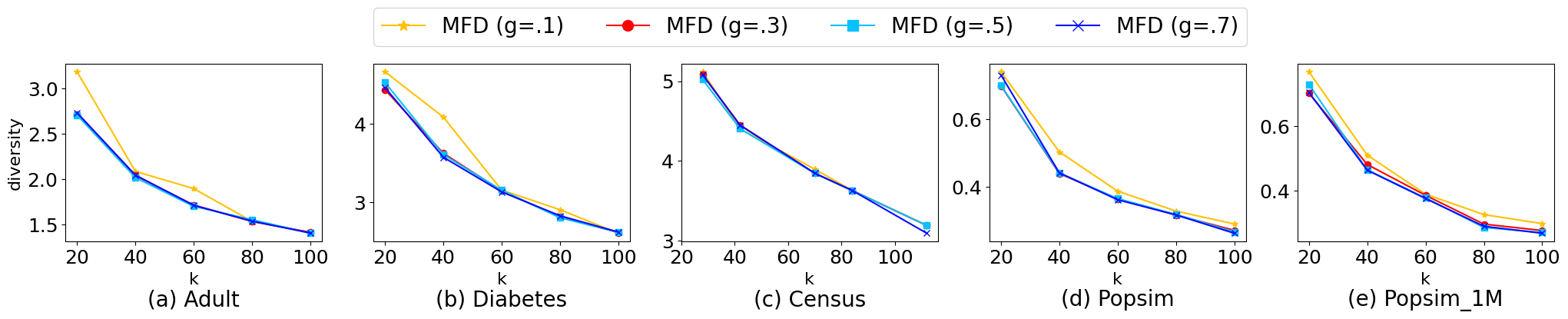}
    \vspace{-5mm}
    \caption{Comparison of diversity vs $k$ for \newAlg{} with different early stopping parameters.\label{fig:diverse:equal:mwu}}
\end{figure*}
\begin{figure*}
    \includegraphics[width=\textwidth]{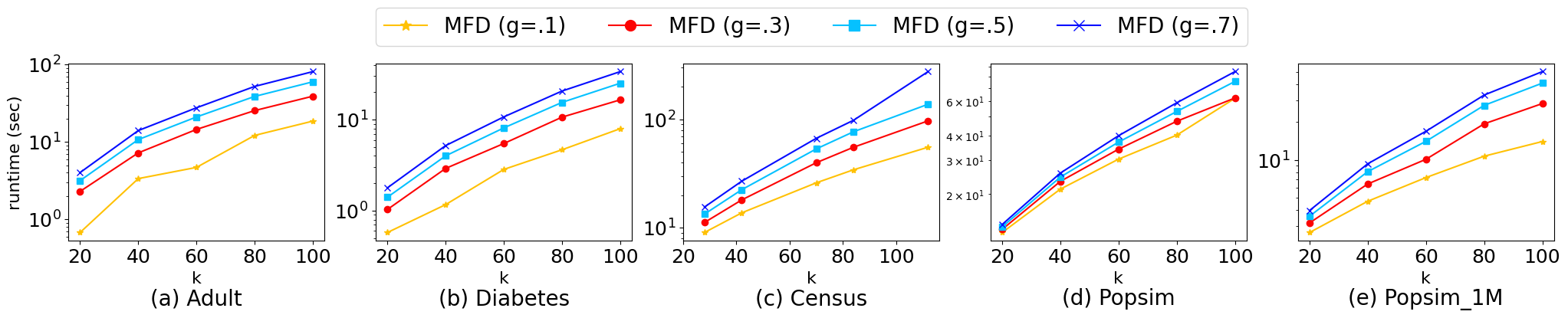}
    \vspace{-5mm}
    \caption{Comparison of running time vs $k$ for \newAlg{} with different early stopping parameters.\label{fig:time:equal:mwu}}
\end{figure*}

\begin{table*}
%\raggedright
\ansc{
\resizebox{\linewidth}{!}{
\begin{tabular}{|c||c|c|c|c||c|c|c|c|c|c|c|c|c|c||c|c|c|c|c|}
\cline{1-9}\cline{11-20}
&\multicolumn{4}{|c||}{$g=0.1$}&\multicolumn{4}{|c|}{$g=0.3$}&&\multicolumn{5}{|c||}{$g=0.1$}&\multicolumn{5}{|c|}{$g=0.3$}\\\cline{1-9}\cline{11-20}
$k$&FN&FY&MN&MY&FN&FY&MN&MY&&Am Ind&As&Afr Am&Nat Haw&Wh&Am Ind&As&Afr Am&Nat Haw&Wh\\\cline{1-9}\cline{11-20}
$20$&$0$&$0$&$0$&$0$&$0$&$0$&$0$&$0$&&$0$&$0.2$&$0.2$&$0.4$&$0.4$&$0.4$&$0.2$&$1$&$0$&$0.2$\\\cline{1-9}\cline{11-20}
$40$&$0.2$&$0$&$1.2$&$0$&$0$&$0$&$0$&$0$&&$0.6$&$1.4$&$1.2$&$0$&$0.4$&$0.2$&$0$&$0.2$&$0$&$0$\\\cline{1-9}\cline{11-20}
$60$&$0$&$0$&$0$&$0$&$0$&$0$&$0$&$0$&&$0.4$&$0.4$&$1.2$&$0$&$0.4$&$0$&$0$&$0.8$&$0$&$0$\\\cline{1-9}\cline{11-20}
$80$&$0$&$0$&$0.4$&$0$&$0$&$0$&$0$&$0$&&$1.4$&$0.4$&$1$&$0$&$0$&$0.6$&$0$&$1.4$&$0$&$0$\\\cline{1-9}\cline{11-20}
$100$&$0$&$0$&$0$&$0$&$0$&$0$&$0$&$0$&&$0$&$0.4$&$0.4$&$1.4$&$0$&$0$&$0$&$0$&$0.8$&$0$\\\cline{1-9}\cline{11-20}
\end{tabular}
}
}
\caption{\ansc{Number of points (on average) per color missed by \newAlg{}-$0.1$ and \newAlg{}-$0.3$ in Diabetes (left) and Popsim (right) datasets.}}
\label{table:misspoints}
\vspace{-2em}
\end{table*}

\noindent \textbf{Datasets}
We consider the following datasets.
%Each dataset is normalized to have zero mean and unit standard deviation, as suggested in~\cite{}.
\begin{enumerate}
    \item The Adult dataset\cite{adult} contains around 32K records of individuals describing their income, and education details. Race and Sex are considered as sensitive attributes to generate $10$ colors (protected groups). 
     \item The Diabetes dataset~\cite{diabetes} contains health statistics of around 100K patients, where Sex and medical prescriptions are used to generate colors.
    \item The Census dataset~\cite{census} contains around 2M records of individuals. Age and Sex are used as sensitive attributes to consider $14$ colors. We choose $6$ numerical attributes to represent the points in the dataset.
    \item Popsim~\cite{nguyen2023popsim} is a semi-synthetic dataset that combined population statistics along with a geo-database. It is used to represent individual-level data with demographic information for the state of Illinois. Race is used as sensitive attribute to consider $5$ colors. The parameters longitude, latitude are used to convert the locations into 4,110,608 points in $\Re^2$.
    \item Popsim\_1M~\cite{nguyen2023popsim}.
    Another version of the same semi-synthetic dataset containing almost 1 million individuals.
    \item We use Beer reviews dataset for experiments in the streaming setting. The dataset contains a large number of reviews over different beers. We categorize the reviews into three groups: reviews related to Lager, Ale, and Other beers.
\end{enumerate}

\noindent \textbf{Baselines}
We compare our algorithm with the state of the art algorithms on the \fairDiv\ problem.

First, we discuss how we implemented our new \newAlg{} algorithm.
We implemented our main algorithm from Corollary~\ref{cor:main} combining the algorithm from Theorem~\ref{thm:main} with the coreset construction in Theorem~\ref{thm:coreset}. For coreset construction, we modified a python implementation for constructing coresets~\cite{vombatkere_github_nodate}. For every color $c_j\in C$, we run the Gonzalez algorithm for $k$ iterations. In the end, we have a coreset of size $m\cdot k$. Then, the coreset is given as input to our \newAlg{} method. We note that we include the coreset construction time in the total running time of \newAlg{}.

Next, we describe some differences between the theoretical \newAlg{} algorithm and our implementation. Instead of constructing a WSPD to run a binary search on the distances, we get as an upper bound $\gamma$ on the maximum possible diversity: We run Gonzalez algorithm for $k$ iterations in the entire point set $P$ without considering the colors. It is known that the minimum pairwise distance among the selected points is an upper bound on the diversity of \fairDiv. If \newAlg{} does not find a feasible solution for diversity $\gamma$, we set $\gamma\leftarrow (1-0.15)\gamma$, and re-run the algorithm, stopping at the first feasible solution.  
Additionally, instead of using a BBD tree we use ParGeo's KD tree~\cite{wang2022pargeo} with modifications to support sum queries. 
In practice, we observe that sometimes the \newAlg{} algorithm runs all $T=O(\eps^{-2}k\log n)$ iterations only to find an unrounded solution which is very similar to ones found by stopping at earlier iterations.  
We modified the algorithm to allow for \emph{early stopping}.  
We introduce a parameter $\stopping \in \left(0, 1\right]$ such that the \newAlg{} algorithm runs for at most $\stopping \cdot T$ iterations rather than the full $T$ iterations from the theory.
We settled on a default value of $\stopping$ in our experiments of $0.3$.
As we will show, this does not affect the quality of the results yet significantly improves the \newAlg{} algorithm's running time.
Since \newAlg{} is a randomized algorithm, in each case, we run the algorithm five times and we report the average diversity and running time.

$\bullet$ \newAlg{}: Our new implementation as described above.

$\bullet$ SFDM-2: It is a streaming algorithm designed and implemented in~\cite{wang2022streaming}
\footnote{The algorithm uses the minimum and maximum pairwise distance to define a range on the diversity. In the original implementation, the authors selected this range manually. In order to be fair with our implementation, we use the same upper bound used in \newAlg{}. As a lower bound we use the minimum non-zero pairwise distance in the coreset of size $m\cdot k$ we construct.}. 
The algorithm uses a parameter $\eps$ to control the error in the solution. We tried two representative errors $\eps=0.15$ and $\eps=0.75$. We call them SFDM-2 ($e=.15$) and SFDM-2 ($e=.75$), respectively.

$\bullet$ FMMD-S: The algorihtm presented and implemented in~\cite{wang2023max}.
%The algorithm creates its own coreset.

$\bullet$ FairFlow~\cite{moumoulidou2021diverse} as implemented in~\cite{wang2023max}.

$\bullet$ FairGreedyFlow~\cite{addanki2022improved} as implemented in~\cite{wang2023max}. %It takes the coreset as input to improve its running time.

All algorithms are implemented in python, except for the kd-tree implementation in ParGeo which is implemented in C++.
If an algorithm takes more than $30$ minutes to finish, we stop its execution and we do not show the results in the figures.
All datasets and our code can be found in~\cite{anonymrepo}.

\paragraph{Setup} We run all our experiments on a e2-standard-16 Google Cloud VM with 16 vCPUs (8 cores) and 64 GB of memory running Debian 11 Bullseye v20231004.
%} on a machine with an Intel i5-7300HQ 4 cores @ 3.50GHz processor and 16GB RAM running Fedora Linux 37 (Workstation Edition) x86\_64.

\subsection{Micro-benchmark experiments}
\label{subsec:micro}
\begin{figure*}
    \includegraphics[width=\textwidth]{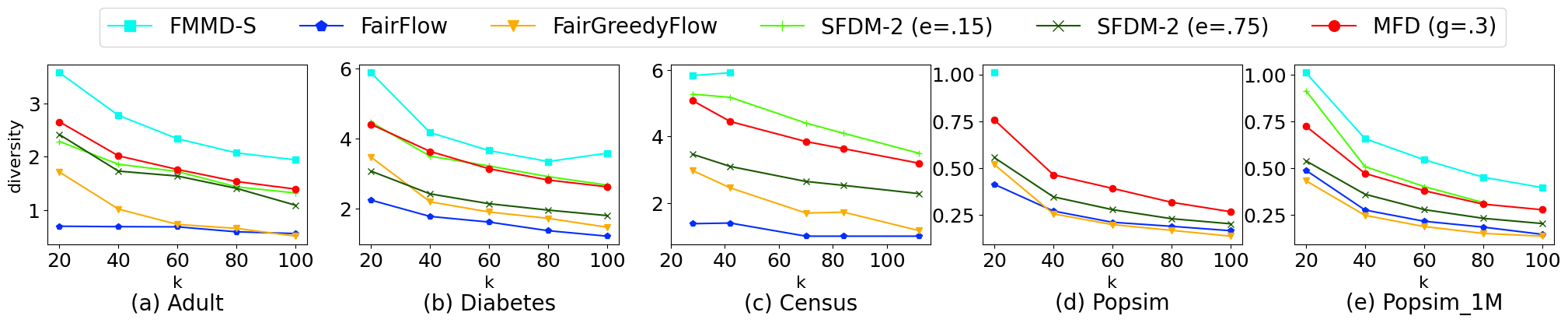}
    \caption{Comparison of diversity vs k for \newAlg{} and baselines. Higher diversity is better. Equal $k_j$.\label{fig:diverse:equal}}
\end{figure*}
\begin{figure*}
    \includegraphics[width=\textwidth]{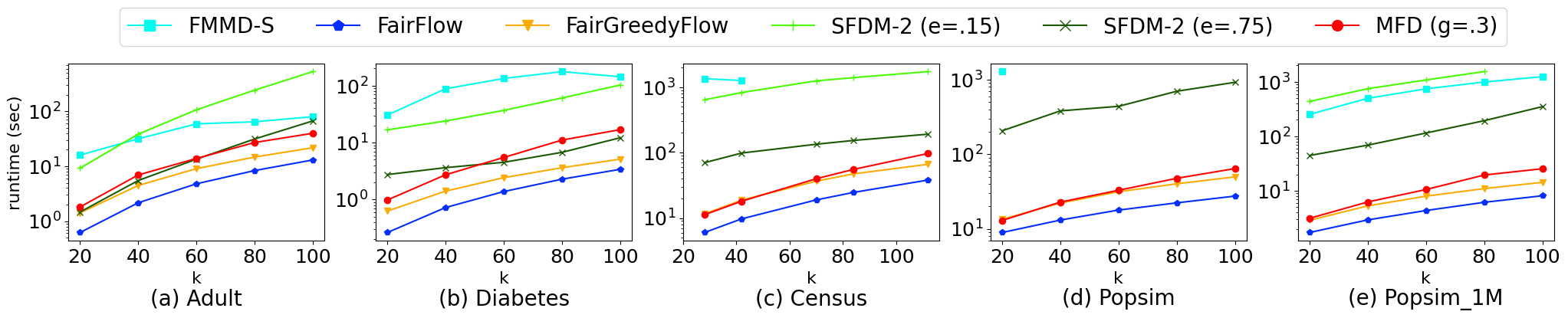}
    \vspace{-5mm}
    \caption{Comparison of running time vs k for \newAlg{} and baselines. Equal $k_j$.\label{fig:time:equal}}
\end{figure*}
\begin{figure*}
    \includegraphics[width=\textwidth]{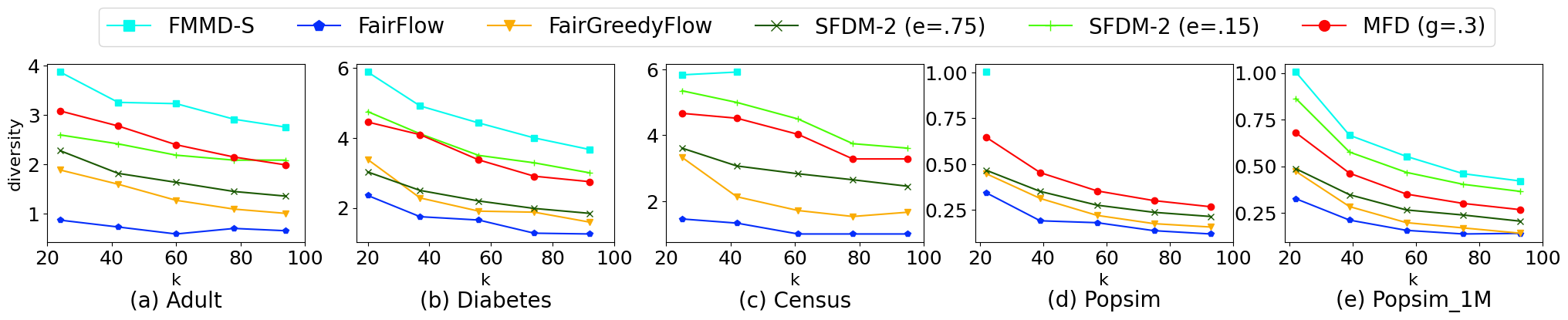}
    \vspace{-7mm}\caption{\label{fig:diverse:prop}Comparison of diversity vs k for \newAlg{} and baselines. Higher diversity is better. Proportional $k_j$.}
\end{figure*}
\begin{figure*}
    \includegraphics[width=\textwidth]{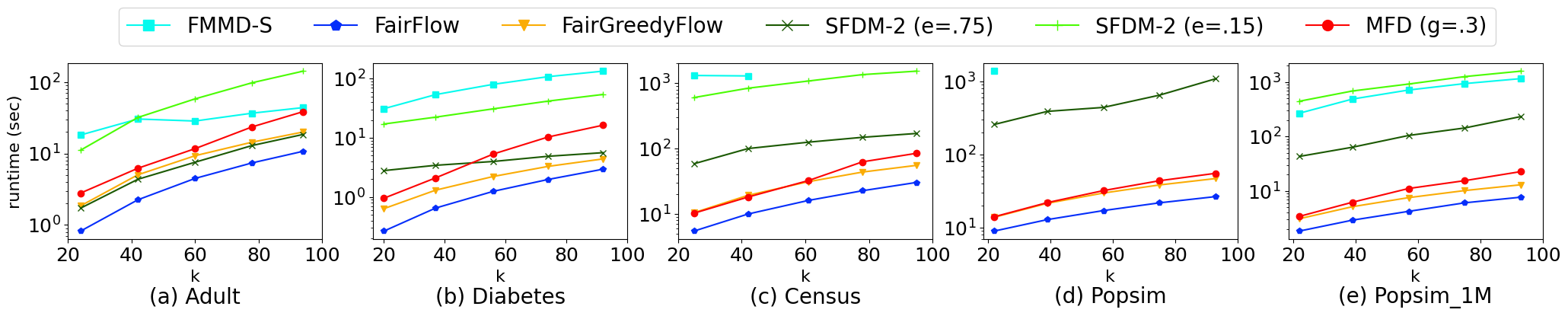}
    \vspace{-7mm}\caption{\label{fig:time:prop}Comparison of running time vs k for \newAlg{} and baselines. Proportional $k_j$.}
\end{figure*}
In this section we compare the diversity and the running time of our \newAlg{} algorithm for different early stopping parameters $\stopping$.
We test $\stopping=0.1, 0.3, 0.5, 0.7$.
In all cases, we set $k_j=k/m$ (equal $k_j$). We also run the same experiments with proportional $k_j=k\frac{|P(c_j)|}{n}$, but we skip them from this version because all observations are identical to the equal case.
In Figure~\ref{fig:diverse:equal:mwu} we show the diversity over different values of $k$ for all datasets. We observe that  the early stopping does not affect the diversity a lot. Figure~\ref{fig:time:equal:mwu} presents the running time over different values of $k$. It shows that the smaller the $\stopping$ is the faster the \newAlg{} is.

Before we conclude that a small value of $\stopping$ is always sufficient, recall that Theorem~\ref{thm:main} (and Corollary~\ref{cor:main}) does not always guarantee at least $k_j$ points in the final solution from every color $c_j\in C$. We show the number of missing points from each color for different parameters $\stopping$.
\ansc{In Table~\ref{table:misspoints} we show two representative results using the Diabetes (left table) and Popsim (right table) dataset for \newAlg{} with $\stopping=0.1$ and $\stopping=0.3$.
For each $k$, we assume that for every color we should take the same number of points, i.e., in Diabetes (Popsim) dataset we should take at least $k/4$ ($k/5$) points per color. We run our algorithm five times and we show the number of missing points on average for each color. The groups in the second row of Diabetes represent sex and Y or N (if meds prescribed), while the groups in the second row of Popsim represent race.
%For every $k$, the left bar shows the real number (as a ratio of $k$) of items we should return from each color, while the right bar contains the number (as a ratio of the solution size) of the points returned by the \newAlg{} algorithm.
When $\stopping=0.1$ \newAlg{} usually misses some points from each group. However for $\stopping=0.3$ in most cases, \newAlg-$0.3$ does not miss any point. For Diabetes, \newAlg-$0.3$ never misses a point. For Popsim, the maximum average number of points it misses for a group is $1.4$. For each $k$ it misses on average $1.16$ points in total, over all groups.
%most of the times, it does not miss any point, and when it does sometimes \newAlg-$0.3$ misses some points however this is always less than $1.4$ on average.
Generally, we observed that in all datasets, for every $\stopping\geq 0.3$ it is rare to miss more than $2$ points.
%In some rare cases we miss a couple of points when $k$ is small (less than $20$) or at most $7$ points in total when $k$ is $100$.
We conclude that \newAlg{} for $\stopping\geq 0.3$ successfully satisfies the fairness requirements.
In the next sections, we fix $\stopping\!=\!0.3$ for \newAlg{}.
}
%As we observe, the smaller the $\stopping$ is the more points we are missing. Interestingly, in all cases \newAlg{} does not miss many points, especially for larger values of $k$. For example, for $\stopping=0.3$, we observed it is rare to miss any points. In some rare cases, we might miss at most $7$ points for some large values of $k$. Interestingly, we only miss some points usually from the majority group.
%\vspace{-2em}

\begin{tcolorbox}[sharp corners, colback=white]
%\begin{Verbatim}[frame=single]
Key Takeaways: 
\newAlg{} with $g=0.3$ satisfies the fairness criterion for most settings, has comparable diversity to other values of $g$ and is highly efficient as compared to higher values of $g$.
%Using smaller $g$ may violate fairness criterion but can be used for increased efficiency.
\end{tcolorbox}

\subsection{End-to-end results}
\label{subsec:endendres}
To compare the quality of \newAlg{} with state-of-the-art, we considered two groups of experiments. In the first group, the fairness constraint ensures that the number of points returned for each color are equal, i.e., $k_j=\frac{k}{m}, \forall j$. In the second group, we choose proportional $k_j$'s, i.e., $k_j=k\frac{|P(c_j)|}{n}$. Generally, the first group of equal $k_j$ leads to more fair solutions that are more difficult to satisfy.

The main goal of this comparison is to identify techniques that maximize diversity within a reasonable amount of time. Two extreme algorithms are:
\textbf{Random selection} would choose $k_j$ points randomly from each color. This approach is expected to be highly efficient but would have very poor diversity. \textbf{Exhaustive search} would exhaustively consider all possible subsets to calculate diversity and return the best solution. This approach may return a highly accurate solution but would be highly inefficient. This approach would not scale to millions sized datasets.

The key goal of this experiment is to identify a technique that returns the best solution within a reasonable amount of time.

\noindent \textbf{Equal number of points from each color.} 
Figure~\ref{fig:diverse:equal} compares the diversity of \newAlg{} and baselines for this setting where equal points from each color are returned (higher diversity is better). We observe that the diversity of the returned solution decreases with increasing $k$. For example, diversity of \newAlg{} reduces from $5$ for $k=20$ to $3.5$ for $k=100$.

 FMMD-S achieves the highest diversity for Adult, Diabetes and Popsim\_1M dataset, but it did not run for Census and Popsim datasets for $k>40$. Therefore, FMMD-S is is not scalable to million scale datasets. Furthermore, it takes at least $50\times$ the time taken by \newAlg{}. Among all other techniques, \newAlg{} achieves the best diversity for most values of $k$ and datasets. SFDM-2 ($\epsilon=0.15$) baseline achieves comparable (or slightly better in some cases) diversity as that of \newAlg{} for datasets and $k$ whenever it ran. However, it did not finish for Popsim dataset for $k\geq 20$ and took at least $50\times$ the time taken by \newAlg{}.
 Among the techniques that ran for all datasets and $k$, \newAlg{} has the best diversity. All other techniques SFDM-2 ($e=0.75$), FairFlow and FairGreedyFlow achieve poorer diversity than \newAlg{}.
 Among these techniques, FairFlow and FairGreedyFlow achieve the lowest diversity across all datasets.

 Figure~\ref{fig:time:equal} compares the running time of different techniques for varying $k$ and datasets. We observe that the running time increases sub-linearly with $k$ for all techniques. FMMD-S and SFDM-2 are generally considerably slower than all other techniques. All other techniques (our approach \newAlg{} and baselines FairFlow, FairGreedyFlow) require less than 40 seconds to identify a diverse set of $100$ points for more than 1M records. In fact, \newAlg{} runs for $k=100$ on 4M records within less than 80 seconds. In contrast to few baselines that take more than $1000$ seconds for a dataset with less than $1$M points. This validates the efficiency of \newAlg{} to identify a fair and diverse solution.

\begin{figure*}
    \includegraphics[width=\textwidth]{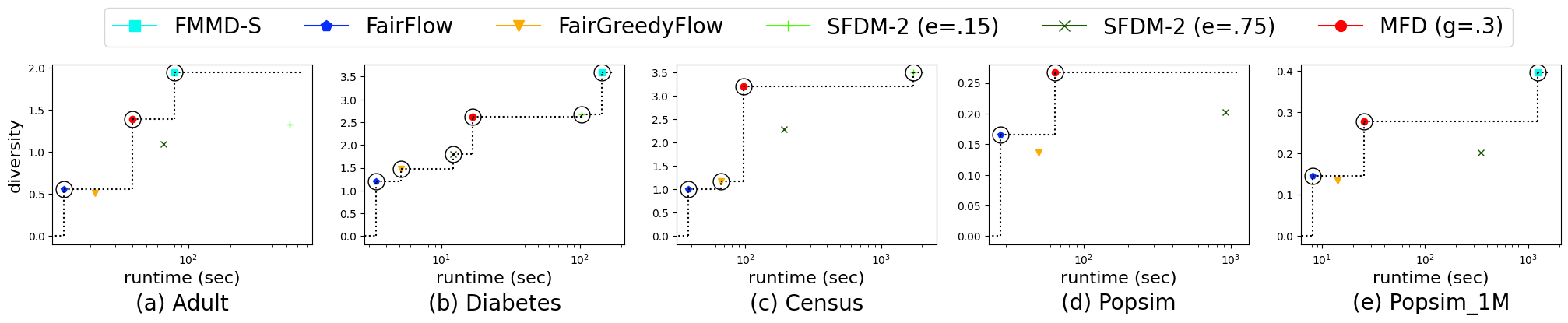}
    \vspace{-6mm}\caption{\label{fig:skyline}\ansb{For every algorithm we show the running time to derive a diverse set for $k=100$ along with the diversity of the returned set. Only \newAlg{} and FairFlow always return a pareto-optimal solution.}}
    \vspace{-1em}
\end{figure*}

\ansb{
Overall, FMMD-S and SFDM-2 return more diverse results, however for some datasets they do not even finish, or they take even 50X more time than our algorithm. At the same time, FairFlow and FairGreedyFlow are faster, but the sets they return have a much worse diversity. Although our algorithm is slightly slower than the fastest algorithms for \fairDiv\ – it scales to large datasets, and returns sets with high diversity in reasonable time. Overall, it achieves the best tradeoff between diversity and time. Every other algorithm is either too slow (no result within 100 sec) or the diversity is too low (5 times worse than the best one).
}

\ansb{In order to further show that \newAlg{} provides the best balance between diversity and running time, in Figure~\ref{fig:skyline}, we fix $k=100$ (similar results hold for any other $k$), and for each algorithm we show the running time to compute a fair and diverse set along with the diversity of the returned set. We represent each algorithm as a point in the (runtime, diversity) plane. An algorithm returns a \emph{pareto-optimal solution} if and only if there is no other algorithm that computes a more diverse (fair) set faster. We observe that \newAlg{} and FairFlow are the only algorithms that always return a pareto-optimal solution. FairFlow always returns a pareto-optimal solution because it is the fastest algorithm, however it returns sets with arbitrarily low diversity (Figure~\ref{fig:diverse:equal} (c)). Overall, \newAlg{} achieves the best equilibrium between diversity and running time.}

\noindent \textbf{Proportional size.} 
Figure~\ref{fig:diverse:prop} and \ref{fig:time:prop} compares the diversity and running time of \newAlg{} and other techniques for the setting where the number of points returned for each color is proportional to their color size in the dataset. The results for this case are similar to that of equal size, where FMMD-S achieves the highest diversity but did not scale to million sized datasets. Similar was the case for SFDM-2 ($\epsilon=0.15$) which performed as good as \newAlg{} but did not finish for Popsim dataset. Among all other baselines, \newAlg{} achieves the best diversity while identifying the solution in less than a minute, even for million sized datasets. 
\begin{tcolorbox}[sharp corners, colback=white]
%\begin{Verbatim}[frame=single]
Key Takeaways:
\begin{enumerate}
 \item  \newAlg{} achieves the best diversity among the techniques that run across all datasets.
    \item \ansc{SFDM-2 ($\epsilon = 0.15$) and FMMD-S achieve higher diversity but they usually take $20\times$ or even $60\times$ more time than \newAlg{} for some datasets.}
    All other baselines achieve worse diversity than \newAlg{} but take almost the same time to run.
    \item \newAlg{} achieves the best quality result within less than a minute for a million scale datasets.
    \item \newAlg{} provides the best balance between diversity and running time.
    \item \ansb{\newAlg{} always returns a pareto-optimal solution.}
    \end{enumerate}
\end{tcolorbox}
\vspace{-0.5em}

\paragraph{Comparison of \newAlg{} and FairGreedyFlow}
Given that the same coreset is given as input, FairGreedyFlow runs in $O(k^2m^4\log(k))$ time, while \newAlg{} runs in $O(k^2m\log^3(k))$ time, so in theory \newAlg{} is faster than FairGreedyFlow.
However, in practice, as shown in Figures~\ref{fig:time:equal},~\ref{fig:time:prop}, FairGreedyFlow runs faster than \newAlg{}.
There are two reasons explaining this.
i) The running time of \newAlg{} actually depends on $1/\eps^{d+2}$ (please check the analysis in Section~\ref{sec:algs}). The parameters $\eps, d$ are constants so we do not include this factor in the final asymptotic complexity. %In theory, \newAlg{} is faster than any previous algorithm when $d$ is a constant.
Our experiments show that for larger $d$ (Diabetes dataset, $d = 8$) FairGreedyFlow is faster than \newAlg{}, whereas for smaller $d$ (Popsim dataset, $d = 2$), the running time is almost identical.
ii) FairGreedyFlow, maps the FairDiv problem to the max-flow problem. In theory, solving an instance of the max-flow is slow. However in practice, they used an optimized implementation of the Ford-Fulkerson algorithm from python’s networkx library.
%While our implementation is fast enough, we do not claim it is fully optimized. More optimizations might be applied to further improve the running time of our implementation.

\subsection{Streaming setting}
\label{subsec:streaming}
Finally, we show experiments for the \streamFairDiv\ problem. We implement our algorithm from Theorem~\ref{thm:streaming}, called \newStreamAlg{}, and compare its efficiency and efficacy with SFDM-2, which is implemented in~\cite{wang2022streaming}.\footnote{Recall that in the offline setting we used the coreset to define a range of diversities for SFDM-2. In the streaming setting, no coreset can be computed before someone visits the entire input set. As described in~\cite{wang2022streaming}, that first introduced SFDM-2, we use the minimum (nonzero) and maximum distance of the entire point set to define the range of diversities (they assume the two quantities are known before the beginning of the stream).}
We do not compare our algorithm with the streaming algorithm proposed in~\cite{addanki2022improved} because i) they did not implement their algorithm, ii)
the update procedure is identical with the update procedure in~\cite{wang2022streaming}, and iii) both algorithms use the same amount of memory, so the conclusions will be almost identical.
%For every color $c_j\in C$, our algorithm uses the doubling algorithm~\cite{charikar1997incremental} to maintain and update a coreset ($k$-center clustering) among the visited points of color $c_j$. In the end, after traversing the full stream, we run our offline algorithm from Theorem~\ref{thm:main} on the coreset to compute the final solution.
At each step of the streaming phase, our algorithm stores $O(km)$ points, while the SFDM-2 stores $O(km\log\Delta)$ points. As we had in the offline case, we run two versions of the SFDM-2 algorithm, the SFDM-2 ($e=.15$) and the SFDM-2 ($e=.75$). For different values of $k$, in Figure~\ref{fig:stream} we show the average update time (average time to insert a new point), the post-processing time (time to construct the final solution after the end of the stream), and the diversity of the sets returned by \newStreamAlg{}, and SFDM-2.
%In all cases \newAlg{} returns sets with diversity similar to the diversity returned by SFDM-2 ($e=.15$), while SFDM-2 ($e=.75$) returns sets with worse (smaller) diversity. \newAlg{} has similar post-processing time with SFDM-2 ($e=.75$), while SFDM-2 ($e=.15$) has on average $10-20\times$ larger post-processing time. Finally, \newAlg{} has the fastest (average) update time.
\newStreamAlg{} has the fastest update time, it has the fastest post-processing time, and it returns sets with diversity close to the diversity of the sets returned by SFDM-2 ($e=.15$). On the other hand, SFDM-2 ($e=.15$) has an expensive update time (sometimes $30\times$ slower than \newStreamAlg{}), while SFDM-2 ($e=.15$) returns sets with very low diversity and it has $2.5\times$ slower update time than \newStreamAlg{}. Overall, \newStreamAlg{} is the best algorithm in the streaming setting because it provides the best balance between update time, post-processing time, and diversity.

\begin{minipage}{\linewidth}
          \begin{figure}[H]\vspace{-1em}
              \includegraphics[width=0.8\linewidth]{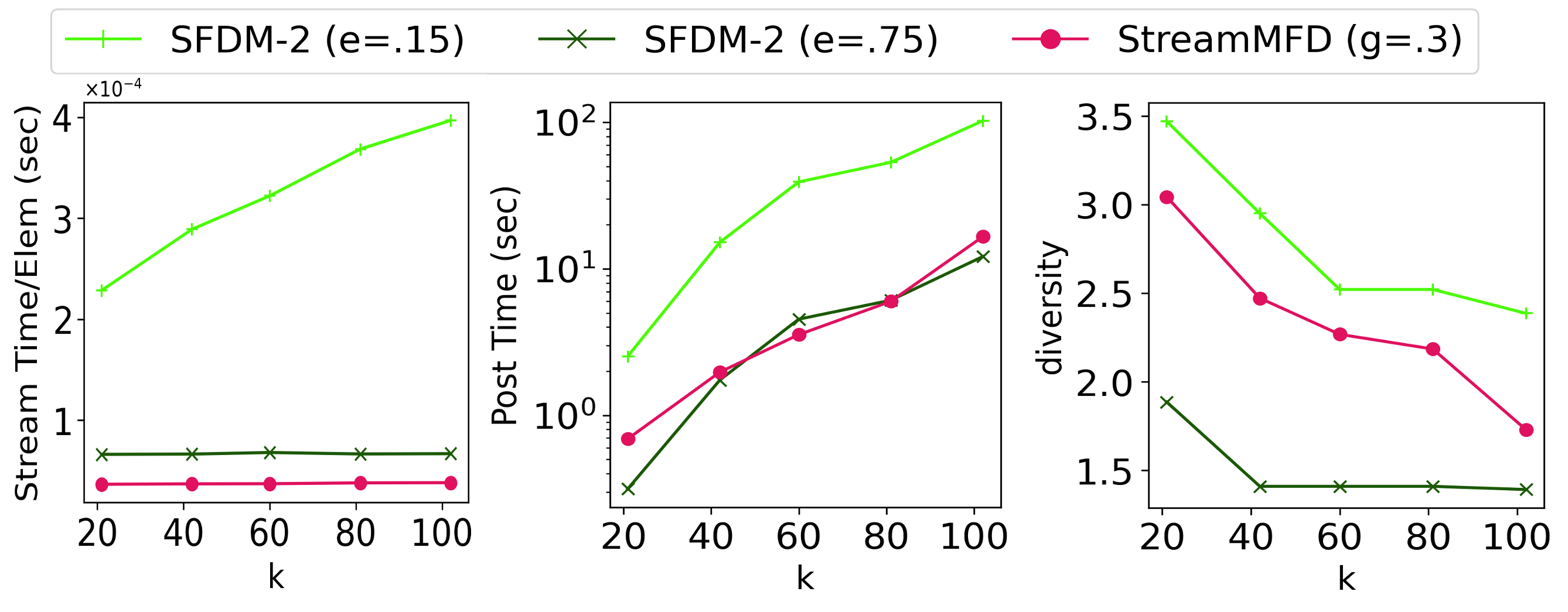}
               \vspace{-1em}\caption{\label{fig:stream}Average update time, post-processing time, and diversity in the streaming setting  for Beer reviews.}
              % \vspace{-1.5em}
          \end{figure}
      \end{minipage}

%\sg{we can move related work to the end}

\section{Related work}
\label{sec:related}
%There is a lot of recent work on the \fairDiv\ problem.
In~\cite{moumoulidou2021diverse} the authors define the \fairDiv\ problem and propose algorithms for general metric spaces. In particular, when $m=2$, they showed an $O(nk)$ time algorithm with $1/4$-approximation factor using linear space.
For any number of colors $m$, they propose FairFlow that runs in $O(kn+m^2k^2\log k)$ time and returns a $\frac{1}{3m-1}$-approximation. When $k$ is small, they also propose a $1/5$-approximation algorithm with exponential running time with respect to $k$.
In \cite{addanki2022improved} they improved some of the results from the previous paper. For any number of colors $m$, they design an LP-based relaxation algorithm to get a $1/2$-approximation satisfying the fairness constraints in expectation, i.e., the expected number of points from color $c_j$ is at least $k_j$. The running time of the algorithm is $O(n^{\lambda})$ it uses $\Omega(n^2)$ space, where $\lambda$ is the exponent to solve an \textsf{LP}.
The same algorithm is extended to return a $1/6$-approximation such that the number of points of a color $c_j$ is at least ${k_j}/{1-\eps}$ with high probability. The running time and space requirement remains the same.
In the same paper they also propose a greedy algorithm, called Fair-Greedy-Flow, that returns an $\frac{1}{(m+1)(1+\eps)}$-approximation in $O(nkm^3)$ time, for constant $\eps$, skipping $\log n$ factors.

In the Euclidean metric, in \cite{addanki2022improved} they constructed a $(1+\eps)$-coreset of size $O(\eps^{-d}mk)$ for the \fairDiv\ problem in $O(nk)$ time.
%In particular, by running the Gonzalez algorithm~\cite{gonzalez1985clustering} for $O(\eps^{-d}k)$ iterations in $P(c_j)$ for every $c_j\in C$, they construct a set $G\subseteq P$ of size $O(\eps^{-d}mk)$ in $O(nk)$ time, such that any $\frac{1}{\beta}$-approximation algorithm for the \fairDiv\ problem executed on $G$ returns a $\frac{1}{(1+\eps)\beta}$-approximation.
%By first constructing the coreset and then running the LP-based algorithm~\cite{addanki2022improved} gets a constant approximation for the \fairDiv\ problem in $O(nk+(mk)^{\lambda})$ time. Similarly, Fair-Greedy-Flow with a coreset runs in $O(nk+k^2m^4)$ time.
Using the coreset, they also design a Fair-Euclidean algorithm that returns a constant approximation in $O(kn+m^{d+2}k\prod_{c_j\in C}k_j^2)$ time. Recently,~\cite{wang2023max} used a coreset construction to propose the FFMD-S algorithm that returns a $\frac{1-\eps}{5}$-approximation in $O(mkn+m^k)$ time for constant $\eps$.
%Their algorithm works for the more general problem: for each color $c_j$ we have both a lower bound and an upper bound on the number of points we have to select, i.e., choose at least $k_j^-$ points and at most $k_j^+$ points from $P(c_j)$.
%Notice that all known constant-approximation algorithms are super-linear on $n$ or/and $k$.

In the streaming setting, the authors in~\cite{wang2022streaming, wang2023fair} presented the SFDM-2 algorithm that stores $O(k\log \Delta)$ items in memory, it takes $O(k\log \Delta)$ time per item for streaming processing, and it requires $O(k^2\log\Delta)$ post-processing time to output a $\frac{1-\eps}{3m+2}$-approximation for the \fairDiv\ problem, for constant $\eps$. In the Euclidean space, the streaming processing time per element can be improved to $O(\log k \cdot \log \Delta)$.
The quantity $\Delta$ is defined as the \emph{spread} of the input set, $\Delta=\frac{\max_{p,q\in P}||p-q||}{\min_{p,q\in P} ||p-q||}=O(2^n)$. In~\cite{addanki2022improved}, they give a streaming algorithm with constant approximation factor, however the space, update time, and post-processing time still depends on $\log\Delta$.
%If this streaming algorithm is used in the offline setting, we would get a $\frac{1-\eps}{3m+2}$-approximation for the \fairDiv\ problem in $O(nk\log\Delta)$ time using $O(n+k\log\Delta)$ space. The spread $\Delta$ can be $2^n$ even for points in $\Re^d$, so the complexities above are linear in the worst case.

Fairness has been studied under different diversity definitions. In fair Max-Sum diversification the goal is to select a set $S\subseteq P$ such that $S$ contains at least $k_j$ points from color $c_j$, and the sum of pairwise distances is maximized, i.e., $\frac{1}{2}\sum_{p,q\in S}||p-q||$ is maximized.
There are a few efficient constant approximation algorithms for this problem~\cite{abbassi2013diversity, borodin2012max, borodin2017max, ceccarello2018fast, ceccarello2020general, cevallos2017local}.
%Abbassi et al.~\cite{abbassi2013diversity} study the Max-Sum problem under matroid constraints (the fair Max-Sum problem can be represented as the Max-Sum problem with matroid constraints). They design a local search algorithm with $(0.5-\eps)$-approximation guarantee.
%Borodin et al.~\cite{borodin2012max, borodin2017max} also studies a bi-criteria optimization problem maximizing the sum of a submodular function and the Max-Sum diversification objective under matroid constraints. They show that the local search algorithm gives also a $(0.5-\eps)$-approximation for this problem. Furthermore,  in~\cite{ceccarello2018fast, ceccarello2020general} the authors constructed a coreset for the Max-Sum diversification problem, and~\cite{cevallos2017local} extended the local search approach to distances of negative type.
The objective function for the fair Max-Sum problem is quite different from \fairDiv, so the techniques from these papers cannot be used in our problem\ansc{,~\cite{addanki2022improved, moumoulidou2021diverse, agarwal2020efficient}.}

%The goal in $k$-center clustering is to find a set of $k$ centers such that the distance form any point to its closest center is minimized.
The Max-Min diversification has a strong connection with the $k$-center clustering. For example, the same greedy Gonzalez~\cite{gonzalez1985clustering} algorithm returns a $\frac{1}{2}$-approximation for the Max-Min diversification~\cite{ravi1994heuristic, tamir1991obnoxious} and $2$-approximation for the $k$-center clustering~\cite{gonzalez1985clustering}.
Kleindessner et al.~\cite{kleindessner2019fair} defined the fair $k$-center problem where the goal is to find $k_j$ centers with color $c_j$ minimizing the $k$-center objective. They proposed a $(3\cdot 2^{m-1}-1)$-approximation algorithm in $O(km^2n+km^4)$ time.
There are many improvements over this algorithm, such as ~\cite{chen2016matroid, chiplunkar2020solve, jones2020fair}.
%leading to a $3$-approximation algorithm for the fair $k$-center that runs in $O(nk)$ time.
The analysis of the algorithms for $k$-center is different than the algorithms for the $k$-Max-Min diversification problem. In some cases, an optimum solution for $k$-center can be arbitrary bad for $k$-Max-Min diversification, \ansc{as shown in~\cite{addanki2022improved}. For example, assume there are two blue points with coordinates $0$ and $5-\eps/2$ and two red points with coordinates $5+\eps/2$ and $10$, for an arbitrary small value $\eps>0$. Selecting the blue point with coordinate $5-\eps/2$ and the red point with coordinate $5+\eps/2$ constructs an optimum solution for the fair $2$-center problem. However the diversity is equal to $\eps$, while the optimum diversity for the \fairDiv\ problem is $10\gg\eps$.}
Hence, it is unclear if algorithms for the fair $k$-center problem can be used for the \fairDiv\ problem.

%Without considering fairness, there is a huge amount of work on diversifying a set of items, using the Max-Min, Max-Sum and other diversity definitions in the offline, streaming, and range query setting~\cite{borassi2019better, agarwal2020efficient, tamir1991obnoxious, ravi1994heuristic, birnbaum2009improved, borodin2017max, cevallos2017local, cevallos2019improved, manzano2016approximation}, however they cannot be extended to solve the \fairDiv\ problem.

%\input{relatedwork}

\section{Future work}
%We studied the \fairDiv\ problem when the input is a set of points in $\Re^d$. We proposed the first near-linear time algorithm with constant approximation ratio for the\fairDiv\ problem. We also showed a more general way to construct a coreset for the \fairDiv\ problem that allows us to accelerate the running time even further. Our coreset construction also  allowed us to design an efficient streaming algorithm that does not depend on the distribution of the data points for the \fairDiv\ in the streaming setting, and the first data structure for the \fairDiv\ in the range-query setting.
There are multiple interesting open problems derived from our work. Is it possible to get a constant approximation for the \fairDiv\ problem satisfying the fairness exactly? The goal is to get at least $k_j$ points instead of at least $\frac{k_j}{1-\eps}$ points, from each color $c_j\in C$. It will also be interesting to check whether the new techniques from this paper can be used to accelerate the other version of the \fairDiv\ problem~\cite{wang2023max} where we have both a lower bound $k_j^-$ and an upper bound $k_j^+$ for each color $c_j\in C$. Finally, it is interesting to study whether other data structures can be used, to get constant approximation in metrics with bounded doubling dimension.

\bibliographystyle{abbrv}
\bibliography{refs}

\end{document}